%% file: HobsonZhuSale.tex
\documentclass[oneside,english]{amsart}
\usepackage[T1]{fontenc}
\usepackage[latin9]{inputenc}
\usepackage{geometry}
\usepackage{amsthm}
\usepackage{amstext}
\usepackage{amssymb}
\usepackage{graphicx}
\usepackage{setspace}
\usepackage{esint}
\usepackage{enumerate}

\newcommand{\sL}{\mathcal L}
\newcommand{\sM}{\mathcal M}

\newcommand{\E}{\mathbb E}
\newcommand{\sE}{\mathcal E}
\newcommand{\dummyq}{q}

\numberwithin{equation}{section}
\numberwithin{figure}{section}

\theoremstyle{plain}
\newtheorem{thm}{\protect\theoremname}
  \theoremstyle{remark}
  \newtheorem*{rem*}{\protect\remarkname}
  \theoremstyle{plain}
  \newtheorem{lem}[thm]{\protect\lemmaname}
  \theoremstyle{plain}
  \newtheorem{prop}[thm]{\protect\propositionname}
  \theoremstyle{plain}
  
   \theoremstyle{plain}
  \newtheorem*{thm*}{\protect\theoremname}
  \newtheorem{res}{\protect\resultname}
  \theoremstyle{plain}
    \theoremstyle{remark}
  \newtheorem{rem}[thm]{\protect\remarkname}

\theoremstyle{plain}
\newtheorem{defn}[thm]{\protect\definitionname}

\makeatother

  \providecommand{\corollaryname}{Corollary}

  \providecommand{\lemmaname}{Lemma}
    \providecommand{\propositionname}{Proposition}
  \providecommand{\remarkname}{Remark}
  \providecommand{\theoremname}{Theorem}
\providecommand{\theoremname}{Theorem}
\providecommand{\resultname}{Result}
\providecommand{\definitionname}{Definition}

\author{David Hobson \and Yeqi Zhu}
\thanks{Department of Statistics, University of Warwick, Coventry, CV4 7AL, UK. D.Hobson@warwick.ac.uk, Yeqi.Zhu@warwick.ac.uk}

\usepackage[english]{babel}

\begin{document}

\title{\textbf{Optimal consumption and sale strategies for a risk averse agent}}
\date{\today}

\begin{abstract}
In this article we consider a special case of an optimal consumption/optimal portfolio problem first
studied by Constantinides and Magill and by Davis and Norman,
in which an agent
with constant relative risk aversion seeks to maximise expected discounted utility
of consumption over the infinite horizon, in a model comprising a risk-free asset and a risky asset
with proportional transaction costs. The special case that we consider is that the cost of purchases
of the risky asset is infinite, or equivalently the risky asset can only be sold and not bought.

In this special setting new solution techniques are available,
and we can make considerable progress towards an analytical solution.
This means we are able to consider the comparative statics of the problem.
There are some surprising conclusions, such as consumption rates are not monotone
increasing in the return of the asset, nor are the certainty equivalent values
of the risky positions monotone in the risk aversion.



\end{abstract}

\maketitle

\section{Introduction}

This article is concerned with the optimal behaviour of an agent whose goal is
to maximise the expected discounted utility of consumption, and who finances consumption
from a combination of initial wealth and from the sale of an initial
endowment of an infinitely divisible security. Her actions are
to choose an optimal consumption strategy and an optimal holding or
portfolio of a risky security, under the restriction that the risky
asset can only be sold, and purchases are not permitted. As such this
problem is a modification of the Merton~\cite{Merton} optimal
consumption/optimal portfolio problem.

Merton~\cite{Merton} considered portfolio optimisation and consumption in a
continuous-time stochastic model, with an investment opportunity set comprising a
risk-free bond and a risky asset with constant return and volatility. Merton
chose to study these issues by first understanding the behaviour of a
single agent acting as a price-taker. Under an assumption of constant relative risk aversion (CRRA) he obtained a closed form solution to the
problem and the optimal strategy in his model consists of trading
continuously in order to keep the fraction of wealth invested in the
risky security equal to a constant.

Merton's model was subsequently generalised
to an incomplete financial market setting where perfect hedging is no
longer possible. Constantinides and Magill~\cite{ConMagill} (see also Constantinides~\cite{Constantinides}) introduced proportional
transaction costs to the model and considered an investor whose aim is
to maximise the expected utility of consumption over an infinite horizon
under power utility. They conjectured the existence of a `no-transaction' region, and that it is optimal to keep the proportion
of wealth invested in the risky asset within some interval.
Subsequently Davis and Norman~\cite{Davis} gave a precise formulation. The Davis and Norman~\cite{Davis} analysis of
the problem is a landmark in the study of transaction
cost problems.

In this article we consider a special case of the Constantanides-Magill-Davis-Norman model in which the transaction costs associated with purchases
of the risky asset are infinite. Effectively purchases are disallowed,
and we may think of an agent who is endowed with a quantity of an asset which she may sell, but which she may not trade dynamically. There are at least two main reasons for considering this special case.
Firstly, there are often situations whereby agents are endowed with units of assets which they may sell but may not repurchase,
whether for legal reasons or because of other trading restrictions. In this sense the problem is interesting in its own right, and as we show,
the model has some counter-intuitive features. Secondly, relative to the Constantanides-Magill-Davis-Norman model new solution techniques
become available which we are able to exploit to give a more complete solution to the problem. With this more complete solution we can investigate the comparative statics of the problem.

\section{Related literature and main conclusions}

\subsection{Related literature}
Davis and Norman~\cite{Davis} were the first to study the Merton model with proportional transaction costs in a mathematically precise formulation.
They showed that under optimal behaviour the no transaction region
is a wedge containing the Merton line and that the optimal buying and
selling
strategies are local times at boundaries chosen to keep the process
inside the
wedge. In the transaction region, transactions take place at infinite
speed and except for the initial transaction, all transactions take
place at the boundaries.
They obtained their results by writing down the (non-linear, second order) Hamilton-Jacobi-Bellman (HJB) equation with free boundary conditions and then by a series of transformations reducing the
problem to one of solving a system of first order ordinary differential
equations. Motivated by Davis and Norman's work, Shreve and
Soner~\cite{SS} studied the same problem but with an approach via
viscosity solutions. They recover the results from Davis and Norman~\cite{Davis}
without imposing all of the conditions of \cite{Davis}.

In related work, Duffie and Sun~\cite{duffie}, Liu~\cite{liu1} and
Korn~\cite{korn} study the problem when there are fixed (as opposed to proportional) transaction
costs. Liu used the HJB approach, deriving an ordinary differential
equation to characterise the value function and solving it numerically.
He found that if there is only a fixed transaction cost, the optimal
trading strategy is to trade to a certain target amount as soon as the
fraction of wealth in stock goes outside a certain range.  Korn~\cite{korn} solved a
similar problem by an impulse control and optimal stopping approach. He
proved the Bellman principle and solved for the reward function by an
iteration procedure under the assumption that the value function is
finite.

Whilst financial assets can be actively traded, in other
contexts dynamic trading is not possible. Svensson and
Werner~\cite{svensson} were the first to consider the problem of
pricing non-traded assets in Merton's model. In the situations we
model, an agent endowed with units of an asset can sell the asset,
but may not make purchases. In the simplest case the agent is
endowed with a single unit of an indivisible asset which cannot
be traded and the problem reduces to an optimal sale problem for
an asset. Evans et al~\cite{evans}, see also Henderson and
Hobson~\cite{HHaap,HHms}, consider an agent with power utility
function who owns an indivisible, non-traded asset and wishes to
choose the optimal time to sell the asset in order to maximise
the expected utility of terminal wealth in an incomplete market.
Their results show that the optimal criterion for the sale of the
asset is to sell the first time the value of the non-traded asset
exceeds a certain proportion of the agent's trading wealth and
this critical threshold is governed by a transcendental equation.
Henderson and Hobson~\cite{david real} also study the problem in
the context of real options, where the investor has a claim on units
of non-traded assets correlated with the risky asset.

\subsection{Informal statement of the main conclusions}
This paper considers an individual who is endowed with cash and units of
an infinitely divisible asset, which can be sold but not dynamically
traded, and who aims to maximise the expected discounted utility of
consumption over an infinite horizon. (The case of an indivisible asset
is considered by Henderson and Hobson~\cite{HHforthcoming}.) The problem
facing the individual is to choose the optimal strategy for the
liquidation of the endowed asset portfolio, and an optimal consumption
process chosen to keep cash wealth non-negative. The price process of
the endowed asset is assumed to follow an exponential Brownian motion
and the agent is assumed to have constant relative risk aversion.

The constraint that the asset can be sold but not bought is equivalent
to an assumption of no transaction costs on sales, and an infinite
transaction cost on purchases. (The assumption of no transaction cost
on sales can easily be relaxed to a proportional transaction cost on sales by working with a process representing
the post-transaction-cost price rather than the pre-cost price.) In
this sense the problem we consider can be interpreted as a special
case of the Davis-Norman problem for Merton's model with transaction
costs in which the transaction cost associated with buying the endowed
asset is infinite.

Our main results are of three types. Firstly we are able to
completely classify the
different types of optimal strategies and the parameter ranges over which they apply. Secondly, we can simplify the
problem of solving for the value function, especially when compared
with direct approaches for solving the HJB equation via smooth fit.
Thirdly, we can perform
comparative statics on quantities of interest, and uncover some
surprising implications of the model.

Some of our main results are as follows.

\begin{res}
If the endowed asset is depreciating over time then the investor should
sell immediately. Conversely, if the mean return is too strong and the
coefficient of relative risk aversion is less than unity, then the
problem is ill-posed, and provided the initial holding of the endowed
asset is positive the value function is infinite.

Otherwise, there are two cases. For small and positive mean return there
exists a finite critical ratio and the optimal sale strategy for
the endowed asset is to sell just enough to keep the ratio of
wealth held in the endowed asset to cash wealth below this critical
ratio. For larger returns it is optimal to first consume all cash
wealth, and once this cash wealth is exhausted to finance
consumption through sales of the endowed asset.
\end{res}

\begin{res}
In the case where the critical ratio is finite then it is given via the
solution of a first crossing problem for a first-order initial-value
ordinary differential equation (ODE). Other quantities of interest can be expressed in terms of the solution of this ODE. In the case where the critical ratio is infinite, the value function can again be expressed in terms of the solution of a first-order ODE.
\end{res}

\begin{res}
We give three sample conclusions from the comparative statics:

\begin{enumerate}
\item The optimal consumption process is not monotone in the drift of the
endowed asset. Although we might expect that the higher the drift, the
more the agent would consume, sometimes the agent's consumption is a
decreasing function of the drift.

\item
The certainty equivalent value of the
holdings of the risky asset is not monotone in risk aversion. For small
quantities of endowed asset, the certainty equivalent value is
increasing in
risk aversion, while for larger quantities, it is decreasing.

\item
The cost of illiquidity (see Definition~\ref{def:cl} below),
representing the loss in welfare of the agent when compared with an
otherwise identical agent who can buy and sell the risky asset with
zero transaction costs, is a U-shaped function of the size of the
endowment in the risky asset.
\end{enumerate}
\end{res}

We work with bond as num\'{e}raire (so that interest rate effects can be ignored) and
then the relevant parameters are the discount parameter and the
relative risk aversion of the agent, and the drift and volatility of
the price process of the risky asset. In the non-degenerate parameter
cases the agent faces a conflict between the incentive to keep a large
holding in the risky asset (since it has a positive return) and the
incentive to sell in order to minimise risk exposure. From the
homothetic property we expect decisions to depend on the ratio between
the value of the holdings of risky asset and cash wealth.

The HJB equation for our problem is second
order, non-linear and subject to value matching and smooth fit of the
first and second derivatives at an unknown free-boundary. One of our key
contributions is to show that the problem can be reduced to a crossing
problem for the solution of a first order ODE. This big simplification
is useful both when considering analytical properties of the solution,
and when trying to construct a solution numerically. We classify the
parameter combinations which lead to different types of solutions and
provide a
thorough analysis of the existence and finiteness of the critical ratio,
and the corresponding optimal strategies.
In the case of a finite and positive critical ratio we show how the
solution to the problem can be characterised by an autonomous
one-dimensional diffusion
process with reflection and its local time.

The structure of the paper is as follows. Firstly, we give a precise
description of the model and then a statement of the main results. The
HJB equation for the problem is second order and non-linear, but
a change of variable makes the equation homogeneous and then a change
of dependent variable reduces the order. Hence the form of the
solution is governed by the solution of a first crossing problem of an
initial value problem for a first order ODE. Even though closed-form
solutions of this ODE are not available we can provide a complete
characterisation of when the first crossing problem has a solution,
and
given
a solution of the first crossing problem we show how to construct the (candidate)
value function. There are two types of degenerate solution (in one
case it is always optimal to liquidate all units of the risky asset
immediately, and in the other the value function is infinite and the
problem is ill-posed). In addition there are two different types of
non-degenerate behaviour (in one case the agent sells units of asset in order to keep the proportion of wealth held in the risky asset below a certain level, and in the other the agent exhausts all her cash reserves before selling any units of the risky asset.) We give proofs of all the main results,
although technical details of the verification arguments are sometimes
relegated to the appendices.

Once the analysis of the problem is complete we are in a position to
consider the comparative statics of the problem.  We consider the
comparative statics of the critical ratio, the value function, the
optimal consumption, the certainty equivalent value of the portfolio
and the cost of illiquidity.

\section{The model and main results}
We work on a filtered probability space
$\left(\Omega,\mathcal{F},\mathbb{P},
\left(\mathcal{F}_{t}\right)_{t\geq0}\right)$
such that the filtration satisfies the usual conditions and is generated
by a standard Brownian motion $B=\left(B_{t}\right)_{t\geq0}$. The
price process $Y=\left(Y_{t}\right)_{t\geq0}$ of the endowed asset
is assumed to be given by
\begin{equation}
Y_{t}=y_{0}\exp\left[\left(\alpha-\frac{\eta^{2}}{2}\right)t
+\eta B_{t}\right],\label{eq:96}
\end{equation}
where $\alpha$ and $\eta>0$ are the constant mean return and volatility
of the non-traded asset, and $y_{0}$ is the initial price.

Let $C=\left(C_{t}\right)_{t\geq0}$ denote the consumption rate of
the individual and let $\Theta=\left(\Theta_{t}\right)_{t \geq 0}$
denote the number of units of the endowed asset held by the
investor. The consumption rate is required to be progressively measurable and
non-negative, and the portfolio process $\Theta$ is progressively measurable,
right-continuous with left limits (RCLL) and non-increasing to
reflect the fact that the non-traded asset is only allowed for
sale. We assume the initial number of shares held by the investor
is $\theta_{0}$. Since we allow for an initial
transaction at time 0 we may have $\Theta_0 < \theta_0$. We write
$\Theta_{0-} = \theta_0$. This is consistent with our convention
that $\Theta$ is right-continuous.

We denote by $X=\left(X_{t}\right)_{t\geq0}$ the wealth process of
the individual, and suppose that the initial wealth is $x_{0}$ where $x_0 \geq 0$. Provided
the only changes to wealth occur from either consumption or from the
sale of the endowed asset, $X$ evolves according to
\begin{equation}
dX_{t}=-C_{t}dt-Y_{t}d\Theta_{t}, \label{eq:95}
\end{equation}
subject to $X_{0-} = x_0$, and $X_0 = x_0 + y_0(\theta_0 -
\Theta_0)$.
We say a consumption/sale strategy pair is admissible if the
components satisfy the requirements listed above and if the
resulting cash wealth process $X$ is non-negative for all time. Let
$\mathcal{A}\left(x_{0},y_{0},\theta_{0}\right)$ denote the set
of
admissible strategies for initial setup
$\left(X_{0-}=x_{0},Y_{0}=y_{0},\Theta_{0-}=\theta_{0}\right)$.

The objective of the agent is to maximise over admissible strategies
the discounted expected utility of consumption over the infinite horizon,
where the discount factor is $\beta$ and the utility function of the agent is assumed
to be CRRA with relative risk aversion $R \in (0,\infty) \setminus {1}$.
In particular, the goal is to find
\begin{equation}
\underset{\left(C,\Theta\right)\in\mathcal{A}(x_{0},y_{0},\theta_{0})}
{\sup}\mathbb{E}
\left[\int_{0}^{\infty}e^{-\beta t}\frac{C_{t}^{1-R}}{1-R}dt\right].
\label{eq:100}
\end{equation}

Since the set-up has a Markovian structure, we expect
the value function, optimal consumption and optimal sale strategy to be
functions of the current wealth and endowment of the agent and of the
price of the risky asset.
Let $V=V(x,y,\theta,t)$ be the forward starting value function
for the problem
so that
\begin{equation} V(x,y,\theta,t) =
\underset{\left(C,\Theta\right)\in\mathcal{A}(x,y,\theta,t)}
{\sup}\mathbb{E}\left[ \left. \int_{t}^{\infty}e^{-\beta
s}\frac{C_{s}^{1-R}}{1-R}ds \right| X_{t-} = x, Y_t=y, \Theta_{t-} =
\theta\right].
\label{eq:100b}
\end{equation}
Here the space of forward starting, admissible strategies
$\mathcal{A}(x,y,\theta,t)$ is such that $C = (C_s)_{s \geq t}$ is
a non-negative progressively measurable process, $\Theta = (\Theta_s)_{s \geq t}$ is
a right-continuous, decreasing and progressively measurable process and satisfies
$\Theta_t - (\Delta \Theta)_t = \theta$, and $X$ given by
$X_s = x - \int_t^s C_u du - \int_{[t,s]} Y_u d \Theta_u$ is
non-negative.

Define the certainty equivalent value (see, for example,~\cite{uip})
$p= p(x,y,\theta,t)$ of the holdings
of
the risky asset to be the solution to
\begin{equation}
V(x+p, y, 0, t) = V(x,y,\theta,t).
\end{equation}
In fact, by the scalings of the problem it will turn out that $p$ is
independent of time (and henceforth we write $p=p(x,y,\theta)$), and depends on the price $y$ of the risky asset and
the quantity $\theta$ of the holdings in the risky asset, only through
the product $y \theta$.

Our goal is to characterise the value function, the optimal consumption
and sale strategies, and the certainty equivalent price $p$.

The key to the form of the solution to the problem is contained in the
following proposition, which concerns the solution of an ODE on $[0,1)$
and which is proved in
Appendix~\ref{app:propertiesofna}.
There is a one-to-one correspondence between the four cases in the
proposition and the four types of solution to the optimal sale problem.

Let $\epsilon={\alpha}/{\beta}$
and $\delta^{2}={\eta^{2}}/{\beta}$.

\begin{prop}
\label{prop:crossings}

For $\dummyq \in [0,1]$ define $m(\dummyq) =
1-\epsilon\left(1-R\right)\dummyq+\frac{\delta^{2}}{2}R\left(1-R\right)\dummyq^{2}$
and $\ell(\dummyq) = 1 + \left( \frac{\delta^2}{2}-\epsilon \right) (1-R)\dummyq -
\frac{\delta^2}{2}(1-R)^2 \dummyq^2 =
m\left(\dummyq\right)+\dummyq\left(1-\dummyq\right) \frac{\delta^{2}}{2}\left(1-R\right)$. Let
$n = n(\dummyq)$ solve
\begin{equation}
\label{eqn:node}
\frac{n' \left(\dummyq\right)}{n\left(\dummyq\right)}
=\frac{1-R}{R\left(1-\dummyq\right)}
-\frac{\delta^{2}}{2}\frac{\left(1-R\right)^{2}}{R}
\frac{\dummyq}{\ell \left(\dummyq\right)-n\left(\dummyq\right)}
\end{equation}
subject to $n(0)=1$ and $\frac{n'(0)}{1 - R}< \frac{\ell'(0)}{1 - R} = \frac{\delta^2}{2} - \epsilon$.
Suppose that if $n$ hits zero, then $0$ is absorbing
for $n$. See Figure~\ref{fig:n}.

For $R<1$, let $\dummyq^* = \inf \{ \dummyq> 0 : n(\dummyq) \leq m(\dummyq) \}$.
For $R>1$, let $\dummyq^* = \inf \{ \dummyq> 0 : n(\dummyq) \geq m(\dummyq) \}$.
For $j \in \{\ell,m,n \}$
let $\dummyq_j = \inf \{ \dummyq> 0 : j(\dummyq) = 0 \} \wedge 1$.

\begin{enumerate}
\item Suppose $\epsilon \leq 0$. Then $\dummyq^* = 0$.

\item Suppose $0 < \epsilon < \delta^2 R$ and if $R<1$, suppose in addition that $\epsilon < \frac{\delta^2}{2} R + \frac{1}{1-R}$. Then $0< \dummyq^* < 1$.

\item Suppose $\epsilon \geq \delta^2 R$ and if $R<1$, $\epsilon < \frac{\delta^2}{2} R + \frac{1}{1-R}$. Then $\dummyq^* = 1 = \dummyq_\ell = \dummyq_n = \dummyq_m$.

\item Suppose $R<1$ and $\epsilon > \frac{\delta^2}{2} R + \frac{1}{1-R}$. Then
$\dummyq_m< \dummyq_n=\dummyq_\ell<1$.
If $R<1$, $\epsilon = \frac{\delta^2}{2} R +
\frac{1}{1-R}$ and $\epsilon < \delta^2 R$ then $\dummyq_m< \dummyq_n=\dummyq_\ell=1$.
If $R<1$, $\epsilon = \frac{\delta^2}{2} R +
\frac{1}{1-R}$ and $\epsilon \geq \delta^2 R$ then $\dummyq^* = 1 = \dummyq_\ell = \dummyq_n = \dummyq_m$.
\end{enumerate}
\end{prop}

\begin{figure}
\begin{centering}
\includegraphics[scale=0.295]{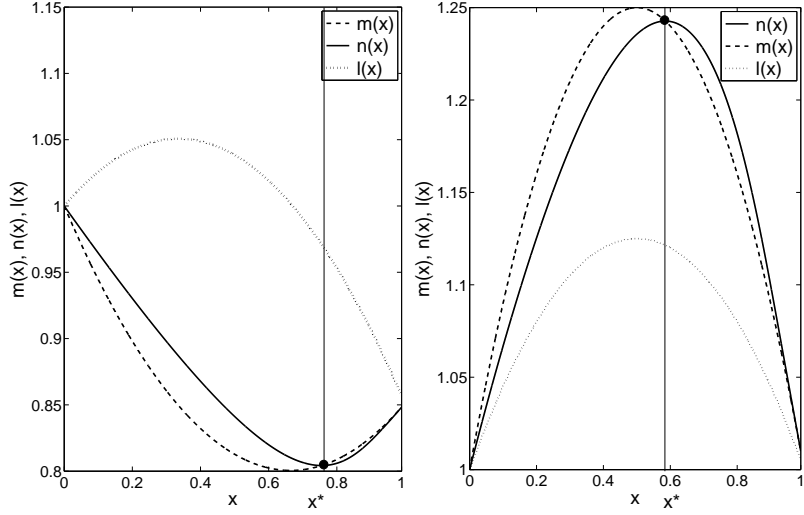}
\par\end{centering}
\caption{Stylised plot of $m(\dummyq)$, $n(\dummyq)$, $\ell(\dummyq)$
and $\dummyq^{*}$. Parameters are chosen to satisfy the conditions in
the second case of Proposition 1 so that $\dummyq^{*} \in (0,1)$. The
left figure is in the case $R<1$ and the right figure $R>1$.}
\label{fig:n}
\end{figure}

\begin{rem}
Note that the condition $\epsilon<\delta^2 R$ is equivalent to $(1-R)m'(1) > 0$. Further, if $R<1$,
then the condition $\epsilon<\frac{\delta^2}{2}R + \frac{1}{1 - R}$ is equivalent to $m(1)>0$.
Also, $n$ has a turning point at $\dummyq^*<1$ if and only if $n(\dummyq^*) = m(\dummyq^*)$. See
Figure~\ref{fig:n}.
In particular, if $m$ is monotone (and $\epsilon>0$) then $\dummyq^* = 1$. Then, if $R<1$, $0<\epsilon<\delta^2 R$ and
$\epsilon<\frac{\delta^2}{2} R + \frac{1}{1-R}$, we have $\dummyq_{\ell} = \dummyq_{n} = 1$.

\end{rem}

\begin{rem}
It is easy to see that $(1-R)n$ is decreasing in $\epsilon$. In fact it can also be shown that over parameter ranges where $0<q^*<1$ then
$q^*$ is increasing in $\epsilon$.
\end{rem}

Theorem~\ref{thm:4cases} divides the parameter space into
the four distinct regions. In particular, it distinguishes the degenerate
cases, and it gives necessary and sufficient conditions for the two
different regimes in the non-degenerate case.

\begin{thm}
\label{thm:4cases}

\begin{enumerate}

\item Suppose $\epsilon \leq 0$. Then it is always optimal to sell the
entire holding of the endowed asset immediately, so that
$\Theta_{t}=0$ for $t\geq0$. The value function for the problem is
$V\left(x,y,\theta,t\right) = (R/\beta)^R e^{-\beta t} (x + y\theta)^{1-R}/1-R$; and
the certainty equivalent value of the holdings of the asset is
$p(x,y,\theta)=y \theta$.

\item Suppose $0 < \epsilon < \delta^2 R$ and $\epsilon <
\frac{\delta^2}{2} R + \frac{1}{1-R}$ if $R<1$. Then there exists a
positive and finite critical ratio $z^*$ and the optimal behaviour is
to sell the smallest possible quantity of the risky asset which is
sufficient to keep the ratio of wealth in the risky asset to cash
wealth at or below the critical ratio. If $\theta>0$ then $p(x,y,\theta) > y
\theta$.

\item
Suppose $\epsilon \geq \delta^2 R$ and $\epsilon < \frac{\delta^2}{2} R + \frac{1}{1-R}$
if $R<1$. Then the optimal consumption and sale
strategy is first to consume liquid (cash) wealth, and then when this
liquid
wealth is exhausted, to finance further consumption from sales of the
illiquid asset. If $\theta>0$ then
$p(x,y,\theta) > y \theta$.

\item Suppose $R<1$ and $\epsilon \geq \frac{\delta^2}{2} R + \frac{1}{1-R}$.
Then the problem is degenerate, and provided $\theta$ is positive, the
value function $V=V(x,y,\theta,t)$ is infinite. There is no unique
optimal strategy, and the certainty equivalent value $p$ is not defined.
\end{enumerate}
\end{thm}

\begin{rem}

In light of Proposition~\ref{prop:crossings} there is one fewer case
for $R>1$. The fourth case in the theorem does
not happen for $R>1$ since the value function is always finite, as in
Merton's problem.

Similarly, when $R<1$, if $\delta^2 \geq 2/(R(1-R))$ then the third case above does not happen.
In that case, as $\epsilon$ increases we move directly from $\epsilon < \frac{\delta^2}{2} R + \frac{1}{1-R}$
and a finite value function and $z^*$ to $\epsilon \geq \frac{\delta^2}{2} R + \frac{1}{1-R}$ and an infinite value function.

\end{rem}

The second and third cases above are non-degenerate and they are further
characterised in Theorem~\ref{thm:maincase} and Theorem~\ref{thm:2ndmaincase}.
In Theorem~\ref{thm:maincase} the
solution is expressed in terms of a one-dimensional autonomous
reflecting stochastic process $J$ and its local time at zero $L$, see
(\ref{eq:Jdef}).

For $0 \leq \dummyq \leq q^*$ define $N(\dummyq) = n(\dummyq)^{-R} (1-\dummyq)^{R-1}$ where $n$ is the solution to (\ref{eqn:node}).
Assuming that $N$ is monotonic, let $W$ be inverse to $N$. Let $h^* = N(\dummyq^*)$. Then
$W(h^*)=\dummyq^*$, and $h^{*} (1- \dummyq^*)^{1-R}  = m(\dummyq^*)^{-R}$.

\begin{thm}
\label{thm:maincase}
\begin{enumerate}[i)]
\item
Suppose $R<1$. Suppose $0 < \epsilon < \delta^2 R$ and $\epsilon <
\frac{\delta^2}{2} R + \frac{1}{1-R}$
so that $0 < \dummyq^* < 1$.
Then $N$ as defined above is increasing, and $W$ is well defined.

Let $z^{*}$ be given by
\begin{equation}
z^{*}=(1-\dummyq^*)^{-1}-1 = \frac{\dummyq^*}{1-\dummyq^*} \in
(0,\infty). \label{eq:886}
\end{equation}
On $\left[1,h^{*}\right]$ let $h$ be the solution of
\begin{equation}
u^{*}-u=\int_{h}^{h^{*}}\frac{1}{(1-R)fW\left(f\right)}df,\label{eq:885}
\end{equation}
where $u^{*}= \ln {z^{*}}$. Let $g$ be given by
\begin{equation}
g\left(z\right)=\begin{cases}
\begin{array}{l}
\left(\frac{R}{\beta}\right)^{R} m(\dummyq^*)^{-R} \left(1+z\right)^{1-R}\\
\left(\frac{R}{\beta}\right)^{R}h\left(\ln z\right)
\end{array} & \begin{array}{l}
\qquad z \in [z^{*}, \infty);\\
\qquad z \in (0, z^{*}].
\end{array}\end{cases}\label{eq:884}
\end{equation}
Then, the value function $V$ is given by
\begin{equation}
V\left(x,y,\theta,t\right)
 =  e^{-\beta t}\frac{x^{1-R}}{1-R}g\left(\frac{y\theta}{x}\right), \hspace{10mm} x>0, \theta>0
\label{eq:883}
\end{equation}
and we can extend this to $x =0$ and $\theta = 0$ by continuity to give
\begin{eqnarray}
\label{eqn:Vdeftheta=0}
V(x,y,0,t) & = & e^{-\beta t}\frac{x^{1-R}}{1-R} \left( \frac{R}{\beta} \right)^R \\
\label{eqn:Vdefx=0}
V(0,y,\theta,t) & = & e^{-\beta t}\frac{y^{1-R}\theta^{1-R}}{1-R} \left( \frac{R}{\beta} \right)^R m(\dummyq^*)^{-R}
\end{eqnarray}

Fix $z_0= y_0 \theta_0 /x_0$. Let $(J,L) = (J_t, L_t)_{t \geq 0}$ be the unique pair such that
\begin{enumerate}
\item $J$ is positive,

\item $L$ is increasing, continuous, $L_0 = 0$, and $dL_t$ is
carried by the set $ \left\{ t:J_{t}= 0 \right\}$,

\item $J$ solves
\begin{equation}
\label{eq:Jdef}
 J_t = (z^*-z_0)^+ - \int_0^t \tilde{\Lambda}(J_s) ds  -
\int_0^t
\tilde{\Gamma}(J_s) dB_s + L_t,
\end{equation}
where $\Lambda(z) = \alpha z + z \left( g(z) - \frac{1}{1-R} z
g'(z)\right)^{-1/R}$, $\Gamma(z)=\eta z$, $\tilde{\Lambda}(j) = \Lambda(z^*-j)$ and
$\tilde{\Gamma}(j) = \Gamma(z^*-j)$.
\end{enumerate}

For such a pair $0 \leq J_t \leq z^*$.

 If $z_0 \leq z^*$ then set $\Theta^*_0 =
\theta_0$ and $X^*_0 = x_0$; else if $z_0 > z^*$ then set
\[ \Theta^*_0 = \theta_0 \frac{z^*}{(1+z^*)}\frac{(1+z_0)}{z_0} \]
and $X^*_0 = x_0 + y_0(\theta_0 - \Theta_0)$.
This corresponds to the sale of
a positive
quantity $\theta_0 - \Theta_0$ of units of the endowed asset at time $0$.

Then, the optimal holdings $\Theta^*_t$
of the endowed asset, the optimal consumption process
$C^*_t = C(X^*_t,Y_t,\Theta^*_t)$, the resulting wealth process
and the certainty equivalent value are given by
\begin{eqnarray}
\Theta^*_t  & = & \Theta_0^* \exp\left\{
-\frac{1}{z^*(1+z^*)} L_t \right\} ; \label{theta*def} \\
X^*_t & = & \frac{Y_t \Theta^*_t}{ (z^* - J_t)}; \label{x*def}\\
C(x,y,\theta) & = & x \left[g\left( \frac{y\theta}{x} \right)
-\frac{1}{1-R} \frac{y\theta}{x}
g' \left( \frac{y\theta}{x} \right)\right]^{-\frac{1}{R}} \label{Cdef}; \\
p(x,y,\theta) & = & x \left[\frac{g\left(\frac{y\theta}{x}\right)}{g(0)}\right]^{\frac{1}{1-R}} - x. \label{pdef}
\end{eqnarray}

\item Now suppose $R>1$ and $0<\epsilon<\delta^2 R$ so that $0 <
\dummyq^* < 1$.
Let all quantities be defined as before.
Then $N$ is decreasing.
On $(h^*, 1)$ $h$ is defined via
\[
u^{*}-u=\int_{h^*}^{h}\frac{1}{(R-1)fW\left(f\right)}df.
\]
The value function, the optimal holdings $\Theta^*$,
the optimal consumption process
$C^*$, the resulting wealth process $X^*$ and the certainty equivalent value $p$ are the same as before.
\end{enumerate}
\end{thm}

\begin{rem}
Recall
that $n$ solves the first order differential equation (\ref{eqn:node}),
and $\dummyq^* \in (0,1)$ is the solution of a first crossing problem
for
$n$. Once we have constructed $n$ and determined $\dummyq^*$, numerically if
appropriate, expressions for all other quantities can be derived by
solving a further integral equation, which can be re-expressed as a
first order differential equation. This two-stage procedure is
significantly simpler than solving the HJB equation directly, as this
equation is second order and non-linear, and subject to second-order
smooth fit at an unknown free boundary.
\end{rem}

\begin{rem}
\label{rem:mertonline}
In the corresponding Merton problem for the unconstrained agent who may both buy and sell the risky asset at zero transaction cost, optimal behaviour for the agent is to hold a fixed proportion $q^M=\alpha/\eta^2 R = \epsilon/\delta^2 R$ of total wealth in the risky asset. This corresponds to keeping $Q_t = Y_t \Theta_t/(X_t + Y_t \Theta_t) = q^M$ or equivalently $Z_t = Y_t \Theta_t/X_t = z^M := q^M/(1-q^M) = \epsilon/(\delta^2 R - \epsilon)$. In Lemma~\ref{lem:n} below we show that $q^* \geq \epsilon/\delta^2 R = q^M$ so that optimal behaviour for the agent who cannot buy units of the risky asset is to keep the ratio of money invested in the risky asset to cash wealth in in interval $[0,q^*]$ where $q^M \in (0,q^*)$.
\end{rem}

The following theorem characterises the solution to the
problem in the second non-degenerate case (the third case in
Theorem~\ref{thm:4cases}).
In this case, the optimal strategy is
to first hold the endowed asset and finance consumption with initial
wealth. When liquid wealth is exhausted, consumption is further
financed by the sale of endowed asset. Here, the critical threshold
$z^* = \infty$.

\begin{thm}
\label{thm:2ndmaincase}

Suppose $\epsilon \geq \delta^2 R$ and if $R<1$,
$\epsilon < \frac{\delta^2}{2} R + \frac{1}{1-R}$.

Let $n$ solve (\ref{eqn:node}) on $[0,1]$. Then for the given parameter combinations we have $q^*=1$. As in Theorem~\ref{thm:maincase}, let $N(q) = n(q)^{-R}(1-q)^{R-1}$. Then $N$ is monotonic.

Let $W$ be inverse to $N$. For $R<1$ define $\gamma:(1,\infty) \mapsto \mathbb R$ by
\begin{equation}
\label{eqn:gammadefR<1}
\gamma(v) = \frac{\ln v}{1-R} + \frac{R}{1-R} \ln m(1) - \frac{1}{1-R} \int_v^\infty \frac{(1-W(s))}{s W(s)} ds .
\end{equation}
If $R>1$ define $\gamma:(0,1) \mapsto \mathbb R$ by
\begin{equation}
\label{eqn:gammadefR>1}
\gamma(v) = -\frac{\ln v}{R-1} - \frac{R}{R-1} \ln m(1) - \frac{1}{R-1} \int_0^v \frac{(1-W(s))}{s W(s)} ds .
\end{equation}
Let $h$ be inverse to $\gamma$ and let $g(z) = (R/\beta)^R h (\ln z)$.

Then, the value function $V$ is given by
\begin{equation}
V(x,y,\theta,t) = e^{-\beta t} \frac{x^{1 - R}}{1 - R}
g\left(\frac{y\theta}{x}\right),  \hspace{10mm} x>0, \theta>0
\label{eq:v}
\end{equation}
which can be extended by continuity to give
\begin{eqnarray}
V(x,y,0,t) & = & e^{-\beta t}\frac{x^{1-R}}{1-R} \left( \frac{R}{\beta} \right)^R, \\
V(0,y,\theta,t) & = & e^{-\beta t}\frac{y^{1-R}\theta^{1-R}}{1-R} \left( \frac{R}{\beta} \right)^R m(1)^{-R} .
\end{eqnarray}

The optimal consumption process $C^*$ is given by $C^*_t = C(X^*_t,Y_t,\Theta^*_t)$ where $C(x,y,\theta)$ is as in (\ref{Cdef})
and the optimal holdings $\Theta^*_t$ of the endowed asset
and the resulting wealth process are given by
\begin{equation}
\Theta^*_t =\begin{cases}
\begin{array}{ll}
\theta _0 & t \leq \tau \\
\theta_0 e^{-\frac{\beta}{R} m(1) (t-\tau)}
& t > \tau
\end{array}\end{cases},
\hspace{8mm}
X^*_t =\begin{cases}
\begin{array}{ll}
x_0 - \int_0^t C(X^*_s, Y_s, \theta_0) ds & t \leq \tau\\
0 &  t > \tau
\end{array}\end{cases},
\label{eq:optimal2}
\end{equation}
where $\tau = \inf \{t \geq 0: X^* _t = 0 \}$. Finally the certainty equivalent value is given by (\ref{pdef}).
\end{thm}

\begin{rem}
Note that $\lim_{z \uparrow \infty} \frac{1}{z} ( g(z) - \frac{z g'(z)}{1-R} )^{-1/R} = \beta m(1)/R$ and hence by continuity we may set $C(0,y,\theta) = y \theta \beta m(1)/R$. Then
for $t > \tau$ we have that
\[ C^*_t = C(0, Y_t, \Theta^*_t) =  \frac{\beta}{R} m(1) Y_t \Theta^*_t . \]
\end{rem}

\section{Proofs and verification arguments}

For $F = F(x,y, \theta, t) \in C^{1,2,1,1}$ such that $F_x>0$ define
operators $\sL$ and $\sM$ by
\begin{eqnarray*}
\mathcal{L}F & = &
\sup_{c > 0} \left\{ e^{-\beta t}\frac{c^{1-R}}{1-R} - c F_x \right\} +
\alpha y F_y + F_t
+ \frac{1}{2} \eta ^2 y^2 F_{yy} \\
& = &
\frac{R}{1 - R} e^{-\frac{\beta}{R} t}F_x ^{1-1/R} +
\alpha y F_y + F_t
+ \frac{1}{2} \eta ^2 y^2 F_{yy}, \\
\mathcal{M}F & = & F_\theta - yF_x.
\end{eqnarray*}

\begin{rem}
The state space of $(X_t,Y_t,\Theta_t,t)$ is $[0,\infty)\times(0,\infty)\times[0,\infty)\times[0,\infty)$, and we want to define $\sL$ and $\sM$ on this region including at the boundary. In practice, all the functions to which we apply the operators are of the form $F(x,y,\theta,t) = e^{-\beta t}\overline{F}(x,y,\theta)$ for some function $\overline{F}$ which is independent of $t$ in which case $F_t = - \beta F$, and this latter form is well defined at $t=0$. Also, we typically need $\sM F$ only for $\theta>0$. Then, given $F$ defined for $x>0$ we can define $F$ at $x=0$ by continuity, and then $\sM F|_{x=0}$ is also well defined. $\sL F$ at $\theta=0$ can be defined similarly, by first defining $F$ at $\theta=0$ by continuity. In order to define $\sL F$ at $x=0$ for $\theta>0$ we extend the domain of $F$ to $x > -\theta y$
and then show that $F_x$ and the other derivatives of $F$ are continuous across $x=0$ with this extension.
\end{rem}

\subsection{The Verification Lemma in the case of a depreciating asset}
Suppose $\epsilon \leq 0$. Our goal is to show that the conclusions of
Theorem~\ref{thm:4cases}(1) hold.

From Proposition~\ref{prop:crossings} we know $q^* = 0$. Define the candidate value function via
\begin{equation}
G(x,y,\theta,t) = e^{-\beta t} \left(\frac{R}{\beta}\right)^R \frac{(x + y\theta)^{1-R}}{1-R} \hspace{10mm} x \geq 0, \theta \geq 0.
\label{eq:vdeg1}
\end{equation}
The candidate optimal strategy is to sell all units of the risky asset immediately. The domain of $G$ can be extended to $-\theta y < x < 0$ for $\theta>0$, using the same functional form as in (\ref{eq:vdeg1}).

Prior to the proof of the theorem, we need the following lemma.

\begin{lem}
\label{deg1:inequal}
Suppose $\epsilon \leq 0$. Consider the candidate value function constructed in (\ref{eq:vdeg1}). Then on $(x \geq 0,\theta>0)$ we have $\mathcal{M} G = 0$, and on $(x \geq 0, \theta \geq 0)$ we have $\mathcal{L}G \leq 0$ with equality at
$\theta = 0$.
\end{lem}
\begin{proof}
Given the form of the candidate value function in (\ref{eq:vdeg1}), we have
\[
\mathcal{M}G = e^{-\beta t} \left(\frac{R}{\beta}\right)^R y (x + y\theta)^{-R} - e^{-\beta t} \left(\frac{R}{\beta}\right)^R y (x + y\theta)^{-R} = 0.
\]
On the other hand, writing $z = y\theta/x$, provided $x>0$
\[
\mathcal{L}G = \beta \left( \frac{R}{\beta} \right)^R e^{-\beta t} \frac{(x + y\theta)^{1-R}}{1-R}
\left[\epsilon (1-R) \frac{z}{1 + z} - \frac{1}{2} \delta^2 R(1 - R) \left(\frac{z}{1 + z}\right)^2 \right]
\leq 0,
\]
with equality at $z = 0$. If $x=0$ then $\sL G = \beta G (1-R)[\epsilon - \frac{\delta^2 R}{2}] < 0$.
\end{proof}

\begin{thm}
\label{thm:deg1}
Suppose $\epsilon \leq 0$. Then the value function is
\begin{equation}
V(x,y,\theta,t) = e^{-\beta t} \left(\frac{R}{\beta}\right)^R \frac{(x + y\theta)^{1-R}}{1-R},
\label{eq:vsimple}
\end{equation}
and the optimal holdings $\Theta^*_t$ of the endowed asset, the
optimal consumption process $C^*_t$ and the resulting wealth process are given by
\begin{equation}
(\bigtriangleup \Theta^{*})_{t = 0} = -\theta_0,
\hspace{6mm}
C^{*}_t = \frac{\beta}{R} (x_0 + y_0\theta_0) e^{-\frac{\beta}{R} t},
\hspace{6mm}
X^*_t = (x_0 + y_0 \theta_0) e^{-\frac{\beta}{R} t}.
\label{eq:optmal3}
\end{equation}
\end{thm}

\begin{proof}

Note that candidate optimal strategy given in (\ref{eq:optmal3}) is to
sell the entire holding of the risky asset at time zero (which gives $X^*_{0} = x_0 + y_0\theta_0$) and thereafter to finance consumption
from liquid wealth, whence the wealth process $(X^*_t)_{t \geq 0}$ is
deterministic and evolves as $dX^*_t = -C^*_t dt$. This gives $X^*_t = (x_0 +
y_0\theta_0)e^{-\frac{\beta}{R}t}$. It follows that the candidate optimal strategy is admissible.

The value function under the strategy proposed in (\ref{eq:optmal3}) is
\begin{eqnarray*}
\mathbb{E} \left[\int_0^\infty e^{-\beta t} \frac{{C^{*}_t}^{1-R}}{1-R} dt \right] & = &
\int_0^\infty e^{-\beta t} \left(\frac{\beta}{R}\right)^{1-R} \frac{\left( e^{-\frac{\beta}{R} t} (x_0
+ y_0\theta_0) \right)^{1-R}}{1-R} dt \\
& = & \left(\frac{R}{\beta}\right)^R \frac{(x_0 + y_0\theta_0)^{1 - R}}{1 - R} = G(x_0,y_0,\theta_0,0).
\end{eqnarray*}
Hence $V \geq G$.

Now, consider general admissible strategies. Suppose first that $R<1$.
Define the process $M=(M_t)_{t \geq 0}$ by
\begin{equation}
M_t = \int_0 ^t e^{-\beta s} \frac{C_s ^{1-R}}{1-R}ds +
G\left(X_t, Y_t, \Theta_t, t\right).
\label{eqn:Mdef}
\end{equation}
Applying the generalised It\^{o}'s formula~\cite[Section 4.7]{gito} to
$M_t$ and suppressing the argument $(X_{s-}, Y_s, \Theta_{s-},s)$ in
derivatives of $G$,
leads to
\begin{eqnarray}
\nonumber M_t - M_0 & = & \int_0^t \left[ e^{-\beta s} \frac{C_s ^{1-R}}{1-R} - C_s G_x + \alpha Y_s G_y
+ \frac{1}{2} \eta^2 Y_s ^2 G_{yy} + G_s \right]ds \\
\nonumber & & + \int_0^t (G_\theta - Y_s G_x) d\Theta_s \\
\label{eqn:Msde} & & + \sum_{\substack{0 \leq s \leq t}}
\left[G(X_s, Y_s, \Theta_s, s) - G(X_{s-}, Y_{s-}, \Theta_{s-}, s) - G_x (\bigtriangleup X)_s
- G_\theta (\bigtriangleup \Theta)_s \right] \\
\nonumber & & + \int_0^t \eta Y_s G_y dB_s \\
\nonumber & = & N_t^1 + N_t^2 + N_t^3 + N_t^4.
\end{eqnarray}
(Note that in the sum we allow for a portfolio rebalancing at $s=0$.)

Lemma~\ref{deg1:inequal} implies that $\mathcal{L}G \leq 0$ and $\mathcal{M}G = 0$, which leads to
$N_t^1 \leq 0$ and $N_t^2 = 0$. Using the fact that $(\Delta X)_s = -Y_s (\Delta \Theta)_s$ and
writing $\theta = \Theta_{s-}$, $x = X_{s-}$, $\chi = - (\Delta
\Theta)_s$ each non-zero jump in $N^3$ is of the form
\[
(\Delta N^3)_s = G(x+y \chi,y,\theta - \chi,s) - G(x,y,\theta,s) +
\chi \left[ G_\theta (x,y,\theta,s) - y G_x (x,y,\theta,s) \right].
\]
Given the form of the candidate value function in (\ref{eq:vdeg1}),
it is easy to see that
$\psi(\phi) = G(x+y \phi,y,\theta - \phi,s)$ is constant in $\phi$,
which gives $\psi(\chi)=\psi(0)$ and $yG_x = G_\theta$ whence
$(\Delta N^3) = 0$.
Then, since $R<1$, we have $0\leq M_t\leq M_0 + N^4_t$, and the local
martingale $N^4_t$ is
bounded from below and hence a supermartingale. Taking expectations
we find
$\mathbb{E} (M_t) \leq M_0 = G(x_0, y_0, \theta_0, 0)$, which gives
\begin{equation}
G(x_0, y_0, \theta_0, 0) \geq \mathbb{E} \int_0^t e^{-\beta s} \frac{{C_s}^{1-R}}{1-R} ds
+ \mathbb{E} G(X_t, Y_t, \Theta_t, t) \geq
\mathbb{E} \int_0^t e^{-\beta s} \frac{{C_s}^{1-R}}{1-R} ds,
\label{eq:16}
\end{equation}
where the last inequality follows since $G(X_t, Y_t, \Theta_t, t) \geq 0$
for $R \in (0,1)$. Letting
$t \to \infty$ in (\ref{eq:16}) leads to
\[
G(x_0, y_0, \theta_0, 0) \geq \mathbb{E} \int_0^\infty e^{-\beta t}
\frac{{C_t}^{1-R}}{1-R} dt,
\]
and taking a supremum over admissible strategies leads to $G \geq V$.

The case $R>1$ is  considered in the Appendix~\ref{app:R>1}.

\end{proof}

\subsection{Proof in the ill-posed case of Theorem~\ref{thm:4cases}}
Recall we are in the case where $R<1$ and $\epsilon
\geq \delta^2 R/2 + 1/(1-R)$.

It is sufficient to give an example of an admissible strategy when $\theta>0$ for which the expected
utility of consumption is infinite. Note that $V(x,y,0,t) = e^{-\beta t}x^{1-R}R^R \beta^{-R}/(1-R)$ so that
the value function is not continuous at $\theta=0$.

Consider a consumption and sale strategy pair
$((\tilde{C})_{t\geq 0}$, $(\tilde{\Theta})_{t\geq 0})$, given by
\begin{equation}
\tilde{\Theta}_t = \tilde{\Theta}_t(\phi) = e^{-\phi t} \theta_0,
\hspace{6mm}
\tilde{C}_t = \tilde{C}_t(\phi) = \phi Y_t \tilde{\Theta}_t = \phi
y_0\theta_0 \exp\left\{\beta (\epsilon -
\delta^2/2
- \phi/\beta)t  + \delta \sqrt{\beta} B_t \right\},
\label{eq:optmal4}
\end{equation}
where $\phi$ is some positive constant.

Note first that that such strategies are admissible
since the corresponding wealth process satisfies
$d\tilde{X}_t = -\phi Y_t \tilde{\Theta}_t dt + Y_t d\tilde{\Theta}_t = 0$,
and hence $(\tilde{X_t})_{t \geq 0} = x_0 > 0$. In particular, consumption is financed
by the sale of the endowed asset only.

The expected discounted utility from consumption $\tilde{G} = \tilde{G}(\phi)$
corresponding to the consumption and sale processes $(\tilde{C}, \tilde{\Theta})$ is
given by
\begin{eqnarray*}
\tilde{G}& = &\mathbb{E}
\left[ \int_0^ \infty e^{-\beta t}  \frac{{\tilde{C}_t}^{1-R}}{1-R} dt \right]  \\
& = & \frac{(\phi y_0 \theta_0)^{1-R}}{1-R} \mathbb{E}
\left[ \int_0^\infty
\exp \left\{\beta\left[(1-R)\left(\epsilon - \frac{\delta^2}{2} - \frac{\phi}{\beta}\right)
- 1\right]t + (1-R) \delta \sqrt{\beta}B_t \right\} dt \right] \\
& = & \frac{(\phi y_0 \theta_0)^{1-R}}{1-R} \int_0^\infty
\exp \left\{\beta(1-R) \left[ \left(\epsilon
- \frac{\delta^2 R}{2} - \frac{1}{1-R} \right) - \frac{\phi}{\beta} \right]
t \right\}  dt
\end{eqnarray*}

Suppose first that $\epsilon > \delta^2 R/2 + 1/(1-R)$. Then for $\lambda \in (0,1)$
and $\phi = \lambda \beta (\epsilon - \delta^2 R/2 - 1/(1-R))$
we have
\[ \left(\epsilon - \frac{\delta^2R}{2} - \frac{1}{1-R}\right) - \frac{\phi}{\beta}
= (1- \lambda) \left(\epsilon - \frac{\delta^2R}{2} - \frac{1}{1-R}\right) > 0,
\]
and $\tilde{G}$ is infinite.

Now suppose that $\epsilon = \delta^2 R/2 + 1/(1-R)$. Then
\[ \tilde{G}(\phi) = \frac{(\phi y_0 \theta_0)^{1-R}}{(1-R)}
\frac{1}{\phi(1-R)} = \phi^{-R}
\frac{(y_0 \theta_0)^{1-R}}{(1-R)^2}
\]
and $\tilde{G}(\phi) \uparrow \infty$ as $\phi \downarrow 0$.

\subsection{The Verification Lemma in the first non-degenerate case with finite
critical exercise ratio.}

Suppose $0 < \epsilon < \delta^2 R$ and if $R<1$, $\epsilon <
\frac{\delta^2}{2}R + \frac{1}{1-R}$. From
Proposition~\ref{prop:crossings} we know $0<\dummyq^*<1$.

Recall the definition $N(\dummyq) = n(\dummyq)^{-R} (1-\dummyq)^{R-1}$ and that $W$
is inverse to $N$. We have $h^* = N(\dummyq^*)$.

\begin{prop}
\label{prop:NWw}
\begin{enumerate}
\item For $R<1$, $N$ is increasing on $[0,\dummyq^*]$.
 $W$ is increasing and $0 < W(v) < \dummyq^*$ on $(1,h^*)$. For $R>1$,
$N$ is decreasing on $[0,\dummyq^*]$.
 $W$ is decreasing and $0 < W(v) < \dummyq^*$ on $(h^*, 1)$.
\item Let $w(v) = v(1-R)W(v)$.
Then $w$
solves
\begin{equation} \frac{\delta^2}{2} w(v) w'(v) -v  + \left(\epsilon -
\frac{\delta^2}{2}
\right) w(v) + \left( v - \frac{w(v)}{1-R} \right)^{1- 1/R}
= 0.
\label{eqn:wode}
\end{equation}
\item For $R<1$ and $1 < v < h^*$, and for $R>1$ and $h^* < v < 1$ we have $w'(v)
< 1 - R w(v)/((1-R)v)$ with $w'(h^*)
= 1 - R w(h^*)/((1-R)h^*)$.
\end{enumerate}
\end{prop}

The proof of Proposition~\ref{prop:NWw} is given in the appendix.

Now define $h$ on $[1,h^*)$ by $\frac{dh}{du} = w(h) = (1-R)hW(h)$ subject to $h(u^*)=h^*$.
Then $h$ solves (\ref{eq:885}) and $w'(h) w(h) = \frac{d^2h}{du^2}$.
Let $g(z) = (\frac{R}{\beta})^R h(\ln z)$. Then $g$ solves (\ref{eq:884}).

\begin{lem}
\label{lem:smoothfit}
Let $m(q^*)^{-R}$, $z^*$ and $g$ be as given in Equations~(\ref{eq:886}) and
(\ref{eq:884})
of
Theorem~\ref{thm:maincase}.
Then,
$g\left(z\right)$, $g'\left(z\right)$, $g''\left(z\right)$
are continuous at $z=z^{*}$.
\end{lem}
\begin{proof}
We have
\[ g(z^*+) = \left(\frac{R}{\beta}\right)^{R}
h^* (1-\dummyq^*)^{1-R} \left(1+z^*\right)^{1-R} =
\left(\frac{R}{\beta}\right)^{R} h^* = \left(\frac{R}{\beta}\right)^{R}
h(u^*) = g(z^*-). \]

For the first derivative we have for $z>z^*$,
\[
zg'(z) =   (1-R) \left( \frac{z g(z)}{1+z} \right)
\]
and then since $\frac{z^*}{1+z^*} = \dummyq^*$, $z^* g'(z^*) = (1-R)
\left(\frac{R}{\beta}\right)^{R} h^* \dummyq^*$.
Meanwhile, for $z<z^*$, and noting that $\frac{dh}{du} = h(1-R)W(h)
=w(h)$,
\[
zg'(z) = \left(\frac{R}{\beta}\right)^{R} h'(u) =
\left(\frac{R}{\beta}\right)^{R} w(h)
\]
so that $z^* g'(z^*-) = \left(\frac{R}{\beta}\right)^{R} w(h^*)$
and the result follows from the substitution $w(h^*) =
(1-R)h^*W(h^*) = (1-R) h^* \dummyq^*$.

Finally, for $z>z^*$
\begin{equation}
\label{eqn:2OSF} z^2 g''(z) = -R(1-R) \left(\frac{R}{\beta}\right)^{R} m(q^*)^{-R} (1+z)^{1-R}
\left(
\frac{z}{1+z} \right)^2 = -R(1-R) g(z) \left(
\frac{z}{1+z} \right)^2
\end{equation}
and $(z^*)^2 g''(z^*+) = -R(1-R)g(z^*) (\dummyq^*)^2$. For $z<z^*$,
\begin{equation}
\label{eqn:2OSFb} z^2 g''(z) = \left(\frac{R}{\beta}\right)^{R} (h'' -
h') =
\left(\frac{R}{\beta}\right)^{R} (w'(h) - 1)w(h)
\end{equation}
and at $z^*$, $(z^*)^2 g''(z^*-) = -R(1-R) \left(\frac{R}{\beta}\right)^{R} h^*
(\dummyq^*)^2$ where we use Proposition~\ref{prop:NWw} (3).
\end{proof}

\begin{prop}
\label{prop:gprop}
Suppose $g\left(z\right)$ solves
(\ref{eq:884}). Then for $R<1$, $g$ is an increasing concave function such
that $g(0)=
(\frac{R}{\beta})^R$. Otherwise, for $R>1$, $g$ is a decreasing convex function
such that $g(0)=(R/\beta)^R$ and $g(z) \geq 0$.
Further, for all values of $R$ we have that
$0 \geq Rg'(z)^2 + (1-R) g(z)
g''(z)$ with equality for $z \geq z^*$.
\end{prop}

\begin{proof}
Consider first $R<1$. Since the statements are immediate in the region $z \geq z^*$, and
since there is second order smooth fit at $z^*$ the result will
follow if $h(-\infty)=1$, $h$ is increasing and, using
(\ref{eqn:2OSFb}),
$w(h) w'(h) - w(h) \leq 0$. The last two properties follow from
Proposition~\ref{prop:NWw} since $w(h)
\geq 0$ and $w'(h) < 1$.

To evaluate $h(-\infty)$ note that
\[ u^* - u = \int_{h(u)}^{h^*} \frac{df}{(1-R)fW(f)}
= \int_{W(h(u))}^{\dummyq^*}
\frac{N'(\dummyq)}{(1-R)N(\dummyq)\dummyq} d\dummyq = \int_{W(h(u))}^{\dummyq^*}
\frac{\frac{\delta^2}{2}(1-R)}{\ell(\dummyq)-n(\dummyq)} d\dummyq . \]
We have that $\ell(\dummyq)-n(\dummyq)$ is bounded away from zero when $\dummyq$ is
bounded away from zero. Further, near $\dummyq=0$ we have
$\ell(\dummyq)-n(\dummyq) \sim C\dummyq$
for some positive constant $C = \ell'(0) - n'(0+)$. Hence
$W(h(-\infty))=0$ and $h(-\infty)=1$, since $W(1)=0$.

For $R>1$, and $z \geq z^*$, the statement holds immediately. For $z \leq z^*$,
 Proposition~\ref{prop:NWw} implies that
$h$ is decreasing and $w(h) \leq 0$, $w'(h) < 1$. Together with (\ref{eqn:2OSFb}),
we have $g$ is a decreasing convex function and $g(z) \geq 0$ given that $h \in [0,1]$.

For the final statement of the proposition, for $z \geq z^*$ the result follows immediately, whereas for $z< z^*$
\[ 
(1-R)g g'' z^2 + R (z g')^2  =  \left( \frac{R}{\beta} \right)^{2R} \left[(1-R) h w(h)[w'(h)-1] + Rw(h)^2 \right] \leq 0
\] 
where the final inequality follows from Proposition~\ref{prop:NWw}(3), noting that $(1-R)w(h) \geq 0$.

\end{proof}

Define the candidate value function via
\begin{equation}
\label{eqn:Gdef}
G(x,y,\theta,t) = e^{-\beta t}\frac{x^{1-R}}{1-R}
g \left( \frac{y
\theta}{x} \right) \hspace{12mm} x>0, \theta>0;
\end{equation}
and extend to $x \leq 0$ and $\theta=0$ using the formulae
\begin{eqnarray}
\label{eqn:Gdefxleq0}
G(x,y,\theta,t) & = & e^{-\beta t}\frac{(x+y\theta)^{1-R}}{1-R}
m(q^*)^{-R} \hspace{10mm} - \theta y < x \leq 0, \theta>0; \\
\label{eqn:Gdeftheta=0}
G(x,y,0,t) & = & e^{-\beta t}\frac{x^{1-R}}{1-R} \left( \frac{R}{\beta} \right)^R \hspace{20mm} x \geq 0, \theta = 0.
\end{eqnarray}

\begin{lem}
\label{lem:concave}
Fix $y$ and $t$. Then $G=G(x,\theta)$ is concave in $x$ and
$\theta$ on $[0,\infty) \times [0,\infty)$. In particular, if $\psi(\chi) = G(x - \chi y \phi, y ,
\theta
+ \chi \phi, t)$, then $\psi$ is concave in $\chi$.
\end{lem}
\begin{proof}
Consider first $R<1$. In order to show the concavity of the candidate value function
it is sufficient
to show that $G(x,0)$ is concave in $x$, $G(0,\theta)$ is concave in $\theta$ and that the Hessian matrix given by
\[
H_{G}=\left(\begin{array}{cc}
G_{xx} & G_{x\theta}\\
G_{x\theta} & G_{\theta\theta}
\end{array}\right).
\]
has a positive determinant, and that one of the diagonal entries is
non-positive. The conditions on $G(x,0)$ and $G(0,\theta)$ are trivial to verify.

Direct computation leads to
\begin{eqnarray*}
G_{xx}\left(x,y,\theta,y\right) & = &
e^{-\beta t}x^{-R-1}\left[-Rg\left(z\right)
+\frac{2R}{1-R}zg'\left(z\right)+\frac{1}{1-R}z^{2}g''
\left(z\right)\right],\\
G_{x\theta}\left(x,y,\theta,t\right) & = &
-e^{-\beta t}x^{-R-1}\frac{y}{1-R}\left[Rg'\left(z\right)
+zg''\left(z\right)\right],\\
G_{\theta\theta}\left(x,y,\theta,t\right) & = &
e^{-\beta t}x^{-R-1}\frac{y^{2}}{1-R}g''\left(z\right),
\end{eqnarray*}
and the determinant of the Hessian matrix is
\begin{equation}
\label{eqn:Hess}
G_{xx}G_{\theta\theta}-\left(G_{x\theta}\right)^{2}
=  -e^{-2\beta t}x^{-2R}\theta^{-2}\frac{R}{(1-R)^2}
\left[(1-R)g\left(z\right)z^{2}g''\left(z\right)
+R\left(zg'\left(z\right)\right)^{2}\right]
\end{equation}
which is non-negative by Proposition~\ref{prop:gprop}.
Further, since $g$ is concave we have that $G_{\theta \theta} \leq
0$.

In order to show the concavity of $\psi$
in $\chi$, it is equivalent to examine the sign of
$\frac{d^2 \psi}
{d \chi^{2}}$. But
\[ \frac{d^2 \psi}
{d \chi^{2}}=\phi^{2}
\left[y^{2}G_{xx}+G_{\theta\theta}-2yG_{x\theta}\right] = \phi^2
(y,1)\det(H_G) (y,1)^T \leq 0. \]

For $R>1$ the argument is similar, except that $G_{\theta \theta} \leq 0$ is now implied by the convexity of $g$.
\end{proof}

\begin{lem}
\label{lem:operator}
Consider the candidate value function constructed in (\ref{eqn:Gdef}).

(a) For $\theta>0$ and $0 \leq x \leq y \theta/z^*$, $\sM G=0$ and $\sL G \leq 0$.

(b) For $\theta>0$ and $x \geq y \theta/z^*$, $\sM G \geq 0$. For $\theta \geq 0$ and $x \geq y \theta/z^*$, $\sL G = 0$.
\end{lem}

\begin{proof}
(a) For $z \geq z^*$, $\sM G=0$ is immediate from the definition of $G$.
For $0 < x \leq y \theta/z^*$
$\sL G$ we have
that $G(x,y,\theta,t) = \left( \frac{R}{\beta} \right)^R
m(q^*)^{-R} e^{-\beta t} \frac{x^{1-R}}{1-R} \left(1+z\right)^{1-R}$ and then
\begin{eqnarray*}
\mathcal{L}G & =& \beta G
\left[m(q^*) - 1
+\epsilon\left(1-R\right)\frac{z}{1+z}
-\frac{1}{2}\delta^{2}R\left(1-R\right)\frac{z^{2}}{\left(1+z\right)^{2}}
\right],
\\
& = & \beta G   \left[ m(\dummyq^*) - m \left( \frac{z}{1+z}
\right)
\right].
\end{eqnarray*}
The required inequality
follows from Part (5) of Lemma~\ref{lem:n} in Appendix~\ref{app:propertiesofna} and
the
fact
that $m(q)/(1-R)$ is
increasing on $(\dummyq^*,1)$. At $x=0$ using both
(\ref{eqn:Gdef}) and (\ref{eqn:Gdefxleq0})
we have $\sL G|_{x=0+} = \sL G|_{x=0-} \beta G [ m(q^*) - m(1)]<0$.

(b) In order to prove $\sL G=0$ for $\theta>0$ we calculate
\begin{eqnarray*}
\sL G(x,y,\theta,t) & = & e^{-\beta t} \frac{x^{1-R}}{1-R}
\left[ R \left( g -
\frac{zg'(z)}{1-R} \right)^{1-1/R} - \beta g + \alpha z g'(z) +
\frac{\eta^2}{2}
z^2 g''(z) \right]
\\
& = & \beta e^{-\beta t} \frac{x^{1-R}}{1-R} \left[ h^{1-1/R} \left( 1 -
\frac{w(h)}{(1-R) h} \right) - h
+ \left( \epsilon - \frac{\delta^2}{2} \right)
w(h) + \frac{\delta^2}{2} w'(h) w(h) \right]
\end{eqnarray*}
and the result follows from Proposition~\ref{prop:NWw}. For $\theta=0$,
$\sL G=0$ is a simple calculation.

Now consider $\sM G$. We have
\begin{equation}
\label{eq:mathcalM}
\mathcal{M}G
=  e^{-\beta
t}x^{-R}y\left[\frac{(1+z)}{1-R}g'
\left(z\right) - g\left(z\right)\right].
\end{equation}
Hence for $R<1$, it is sufficient to show that
$\psi(z) \geq 0$ on $(0,z^*]$ where
\[
\psi\left(z\right)=\frac{1+z}{1-R}-\frac{g\left(z\right)}{g'
\left(z\right)}.
\]
By value matching and smooth fit
$g(z^*) = m(q^*)^{-R} \left(1+z^{*}\right)^{1-R}$ and $z^* g'(z^*) =
m(q^*)^{-R} (1-R)\left(1+z^{*}\right)^{-R}$. Hence $\psi(z^*)=0$ and
it is sufficient to show that $\psi$ is decreasing.
But
\begin{eqnarray}
\psi'\left(z\right)
& = & \frac{R}{1-R}
+\frac{g\left(z\right)g''\left(z\right)}{g'\left(z\right)^{2}}
\nonumber \\
& = &
\frac{R}{1-R}
+\frac{h\left[w\left(h\right)w'\left(h\right)-w\left(h\right)\right]}
 {w\left(h\right)^{2}} \nonumber \\
& \leq & 0 \label{eqn:psi}
\end{eqnarray}
where the last inequality follows from Proposition~\ref{prop:NWw}. Similarly, for $R>1$, provided that $g$ is decreasing by Proposition~\ref{prop:gprop}, it is sufficient to show that
$\psi$ is increasing. But Proposition~\ref{prop:NWw} gives
\[ 
\psi'\left(z\right)
 =  \frac{R}{1-R}
+\frac{g\left(z\right)g''\left(z\right)}{g'\left(z\right)^{2}}
 =
\frac{R}{1-R}
+\frac{h\left[w\left(h\right)w'\left(h\right)-w\left(h\right)\right]}
 {w\left(h\right)^{2}}
\geq  0.
\] 
\end{proof}

\begin{prop}
\label{prop:Z}
Let $X^*$, $\Theta^*$ and $C^*$ be as defined in Theorem~\ref{thm:maincase}.
Then they correspond to an admissible wealth process. Moreover $Z_t^* = Y_t
\Theta^*/X^*_t$ satisfies $0 \leq Z^*_t \leq z^*$.
\end{prop}

\begin{proof}
Note that if $y_0 \theta_0/x_0 > z^*$ then the optimal
strategy includes a sale of the endowed asset at time zero, and the
effect of the sale is to move to new state variables
$(X^*_0,y_0,\Theta^*_0,0)$ with the property that $Z^*_0 = y_0
\Theta^*_0/X^*_0 = z^*$.

Recall the definitions of $\tilde{\Lambda}$ and $\tilde{\Gamma}$ and
set
$\Sigma(z)=z(1+z)$ and $\tilde{\Sigma}(j) = \Sigma(z^*-j)$.

Consider the equation
\begin{equation}
\label{eqn:Jsde} \hat{J}_{t}=\hat{J}_{0}-\int_{0}^{t}
\tilde{\Lambda}\left(\hat{J}_{s}\right)ds
-\int_{0}^{t}\tilde{\Gamma}\left(\hat{J}_{s}\right)dB_{s}
+ \hat{L}_t
\end{equation}
with initial condition $\hat{J}_0 = (z^*- z_0)^+$.
This equation is associated with a stochastic differential equation with reflection
(Revuz and Yor~\cite[p385]{RY}) and has a unique solution $(J,L)$ for which
$(J,L)$ is adapted, $J \geq
0$, $L_0=0$ and $L$ only increases when $J$ is zero.

Note that $\tilde{\Lambda}(z^*) = \Lambda(0)=0=\Gamma(0) = \tilde{\Gamma}(z^*)$ and
hence $J$ is bounded above by $z^*$.

Recall that $\Theta^*_t = \Theta^*_0 \exp( - L_t/ \tilde{\Sigma}(0))$. Then
$\Theta^*_t$
is adapted, continuous and hence progressively measurable (Karatzas and Shreve~\cite[p5]{KaraShre}). $\Theta^*_t$ is also decreasing and $d \Theta^*_t = - \Theta^*_t dL_t/
\tilde{\Sigma}(0) = - \Theta^*_t dL_t/ \tilde{\Sigma}(J_t)$ since $L$ only grows
when $J=0$.

Then let $Z^*_t = z^* - J_t$, $X^*_t = \Theta^*_t Y_t/Z^*_t$ and $C^*_t =
X^*_t(g(Z^*_t) - Z^*_t g'(Z^*_t)/(1-R))^{-1/R}$. Then $X^*$ and $C^*$ are positive
and progressively measurable. It remains to show that $X$ is the wealth process arising from the
consumption and sale strategy $(C^*,\Theta^*)$.
But, from (\ref{eqn:Jsde}) and using, for example $\tilde{\Lambda}(J_t) =
\Lambda(Z^*_t)$,
\[ 
dZ^*_t =
\Lambda\left(Z^*_{t}\right)dt+\Gamma\left(Z^*_{t}\right)dB_{t}
+\Sigma\left(Z^*_{t}\right)\frac{d\Theta_{t}^{*}}{\Theta_{t}^{*}}.
\] 
and then
\begin{eqnarray*}
d X^*_t & = & \frac{\Theta^*_t Y_t}{Z^*_t} \left[ \frac{d \Theta^*_t}{\Theta^*_t} +
\frac{d Y_t}{Y_t} - \frac{d Z^*_t}{Z^*_t} +  \left(\frac{d Z^*_t}{Z^*_t} \right)^2
- \frac{d Y_t}{Y_t} \frac{d Z^*_t}{Z^*_t} \right] \\
& = & X^*_t \left[ \left( \eta - \frac{\Gamma(Z^*_t)}{Z^*_t}\right)
dB_t
+ \left( \alpha - \frac{\Lambda(Z^*_t)}{Z^*_t} +
\frac{\Gamma(Z^*_t)^2}{(Z^*_t)^2}
-
\eta
\frac{\Gamma(Z^*_t)}{Z^*_t}
\right) dt   \right]
+ \left( \frac{Y_t}{Z^*_t} - \frac{Y_t}{Z^*_t} \frac{\Sigma(Z^*_t)}{Z^*_t}
\right) d \Theta^*_t \\
& = & - C^*_t dt  - Y_t d \Theta^*_t
\end{eqnarray*} as required, where we use the definitions of $\Lambda$, $\Gamma$ and
$\Sigma$ for the final equality.
\end{proof}

\begin{proof}[Proof of Theorem \ref{thm:maincase}.]
First we show that there is a strategy such that the
candidate value function is
attained, and hence that $V \geq G$.

Observe first that if $y_0 \theta_0/x_0 > z^*$ then
\[ \theta_0 - \Theta_0^* = \theta_0 \left( 1 - \frac{z^*}{1+z^*}
\frac{1+ z_0}{z_0} \right) \]
and
\[ X^*_0 = x_0 + y_0(\theta_0 - \Theta^*_0) = x_0 \frac{(1+
z_0)}{(1+z^*)} \]
Then, since $g(z^*)/g(z_0) = (1+z^*)^{1-R}/(1+z_0)^{1-R}$ for $z_0>z^*$,
\[ G(X_0^*,y_0,\Theta_0^*,0) =
\frac{(X^*_0)^{1-R}}{1-R} g(z^*) =
 \frac{x_0^{1-R}}{1-R} g(z_0) = G(x_0,y_0,\theta_0,0) .
\]

For a general admissible strategy define the
process $M=\left(M_{t}\right)_{t\geq0}$ by
\begin{equation}
M_{t}=\int_{0}^{t}e^{-\beta
s}\frac{C_{s}^{1-R}}{1-R}ds+
G\left(X_{t},Y_{t},\Theta_{t},t\right).\label{eq:15}
\end{equation}
Write $M^*$ for the corresponding process under the proposed
optimal strategy.
Then $M^*_0 = G(X_0^*,y_0,\Theta_0^*,0) = G(x_0,y_0,\theta_0,0)$ so there is no jump
of $M^*$ at $t=0$.
Further, although the optimal strategy may include the sale of a
positive quantity of the risky asset at time zero, it follows from
Proposition~\ref{prop:Z} that
thereafter the process $\Theta^*$ is continuous and such that
$Z^*_t = Y_t
\Theta^*_t / X^*_t \leq z^*$.

From the form of the candidate value function and the definition of
$g$ given in
(\ref{eq:884}), we know that $G$
is $C^{1,2,1,1}$.
Then applying It\^{o}'s formula to $M_{t}$, using the continuity of $X^*$ and
$\Theta^*$ for $t>0$, and writing
$G_\cdot$
as shorthand
for $G_\cdot(X^*_s,Y_s, \Theta^*_s,s)$
we have
\begin{eqnarray}
M^*_{t}-M_{0} & = & \int_{0}^{t}\left[e^{-\beta
s}\frac{(C_s^*)^{1-R}}{1-R}-C^*_sG_{x}+\alpha
Y_sG_{y}+\frac{1}{2}\eta^{2}Y_s^{2}G_{yy}+G_{t}\right]ds\nonumber \\
&& + \int_{(0,t]}\left(G_{\theta}-Y_sG_{x}\right)d\Theta^*_s
\label{eq:24}
\\ 
&& +  \int_{0}^{t}\eta Y_s G_{y}dB_{s}\nonumber \\
& =: & N_{t}^{1}+N_{t}^{2}+N_{t}^{3} \nonumber   
\end{eqnarray}
Since $Z^*_t \leq z^*$, and since $C^*_t = e^{-\beta s/R}G_x^{-1/R}$ and
$\sL G = 0$ for
$z \leq z^*$ we
have $N_{t}^{1}=0$. Further, $d \Theta_s \neq 0$ if and only if
$Z^*_t = z^*$ and then $\sM G = 0$, so that $N_{t}^{2}=0$.

To complete the proof of the theorem we need the following lemma which
is proved in Appendix~\ref{app:mg}.

\begin{lem}
\label{lem:mg}
\begin{enumerate}
\item $N^3$ given by $N^3_t = \int_0^t
\eta Y_s G_{y}(X^*_s, Y_s, \Theta^*_s,s)
dB_{s}$ is a martingale.
\item $\lim_{t \uparrow \infty} \E[ G( X^*_t, Y_t, \Theta^*_t, t)] = 0$.
\end{enumerate}
\end{lem}

Returning to the proof of the theorem, and taking expectations on both sides of
(\ref{eq:24}), we have
\(
\mathbb{E}\left[M^*_{t}\right]=M_{0},
\)
which leads to
\[ 
G\left(x_{0},y_{0},\theta_{0},0\right)
=\mathbb{E}\left(\int_{0}^{t}e^{-\beta
s}\frac{(C^*_{s})^{*1-R}}{1-R}ds\right)
+\mathbb{E}\left[G\left(X^*_t,y,\Theta^*_t,t\right)\right].
\] 
Using the second part of Lemma~\ref{lem:mg}
and applying the monotone convergence theorem, we
have
\[
G\left(x_{0},y_{0},\theta_{0},0\right)
= \mathbb{E}\left(\int_{0}^{\infty}e^{-\beta
s}\frac{C_{s}^{*1-R}}{1-R}ds\right) \]
and hence $V \geq G$.

Now we consider general admissible strategies.
Applying the generalised It\^{o}'s
formula~\cite[Section 4.7]{gito} to $M_{t}$ leads to
the same expression as in (\ref{eqn:Msde}).
Lemma~\ref{lem:operator} implies that under general admissible
strategies, $N_{t}^{1}\leq0,\; N_{t}^{2}\leq0$.
Consider the jump term,
\begin{equation}
N_{t}^{3}= \sum_{0 \leq s\leq
t}\left[G\left(X_{s},Y_{s},\Theta_{s},s\right)
-G\left(X_{s-},Y_{s},\Theta_{s-},s\right) - G_x (\Delta X)_s
 - G_\theta (\Delta \Theta)_s
\right]
\end{equation}
Using the fact that $(\Delta X)_s = -Y_s (\Delta \Theta)_s$ and
writing $\theta = \Theta_{s-}$, $x = X_{s-}$, $\chi = - (\Delta
\Theta)_s$ each non-zero jump in $N^3$ is of the form
\[
(\Delta N^3)_s = G(x+y \chi,y,\theta - \chi,s) - G(x,y,\theta,s) +
\chi \left[ G_\theta (x,y,\theta,s) - y G_x (x,y,\theta,s) \right].
\]
But, by Lemma~\ref{lem:concave}, $G(x+y \chi,y,\theta - \chi,s)$ is
concave in $\chi$ and hence
$(\Delta N^3) \leq 0$.

For $R<1$ the rest of the proof is exactly as in
Theorem~\ref{thm:deg1}.
The case of $R>1$ is covered in Appendix~\ref{app:R>1}.

\end{proof}

\subsection{The Verification Lemma in the second non-degenerate case with no
finite critical exercise ratio.}

Throughout this section we suppose that $\epsilon \geq \delta^2R$ and that if $R<1$ then $0<\epsilon<\frac{\delta^2}{2}R
+ \frac{1}{1-R}$.
It follows that $\dummyq^{*} = 1$ and $z^{*} = \infty$, and that $n(1)=m(1)>0$.

Recall the definition of $n$ in (\ref{eqn:node}) and the subsequent definitions of $N$ by $N(q)=n(q)^{-R}(1-q)^{R-1}$ and $W=N^{-1}$. Suppose $R<1$ and define $\gamma$ as in (\ref{eqn:gammadefR<1}) by
\[ \gamma(v) = \frac{1}{1-R} \ln v + \frac{R}{1-R} \ln m(1) - \frac{1}{1-R}\int_v^\infty \frac{1- W(s)}{sW(s)} ds. \]
In the case $R>1$ define $\gamma$ via (\ref{eqn:gammadefR>1}) so that
\[ \gamma(v) = -\frac{1}{R-1} \ln v - \frac{R}{R-1} \ln m(1) - \frac{1}{R-1}\int_0^v \frac{1- W(s)}{sW(s)} ds. \]
For all $R$ define also $\tilde{\gamma}$ by
\[\tilde{\gamma}(v) = \frac{\ln v}{1-R} - \gamma(v) .
\]

Let $h$ be inverse to $\gamma$ and set $g(z) = (R/\beta)^R h(\ln z)$.

\begin{lem}
\label{lem:gamma}
\begin{enumerate}
\item Suppose $R<1$. Then $\gamma:(1,\infty) \mapsto (-\infty,\infty)$ is well defined, increasing, continuous and onto.  Furthermore,
\[ \lim_{v \uparrow \infty} \tilde{\gamma}(v) = \frac{-R}{1-R} \ln m(1) \hspace{10mm} \mbox{and} \hspace{10mm} \lim_{v \uparrow \infty} (1-W(v)) e^{\gamma(v)} =1 .\]
Suppose $R>1$. Then $\gamma:(0,1) \mapsto (-\infty,\infty)$ is well defined, decreasing, continuous and onto.  Furthermore,
\[ \lim_{v \downarrow 0} \tilde{\gamma}(v) = \frac{R}{R-1} \ln m(1) \hspace{10mm} \mbox{and} \hspace{10mm} \lim_{v \downarrow 0} (1-W(v)) e^{\gamma(v)} =1 .\]
\item $h$ solves $h'=(1-R)h W(h)$, and $h(-\infty) = 1$.
\end{enumerate}
\end{lem}

\begin{proof}
Suppose $R<1$, the proof for $R>1$ being similar. First we want to show that
\[ \int^\infty \frac{1-W(s)}{sW(s)} ds < \infty, \hspace{10mm} \mbox{and} \hspace{10mm} \int_{1+} \frac{1-W(s)}{sW(s)} ds = \infty, \]
which, given $\lim_{s \uparrow \infty} W(s)=1$ and $\lim_{s \downarrow 1} W(s)=0$ is equivalent to
\[ \int^\infty \frac{1-W(s)}{s} ds < \infty; \hspace{20mm} \int_{1+} \frac{1}{W(s)} ds = \infty. \]
But $(1-q) N(q)^{1/(1-R)} \stackrel{q \uparrow 1}{\longrightarrow} n(1)^{-R/(1-R)}$ and so
$(1-W(s)) \sim n(1)^{-R/(1-R)} s^{-1/(1-R)}$ for large $s$ and the first integral is finite. Conversely, since $N'(0+) = \kappa$ for some $\kappa \in (0,\infty)$ we have $W'(1+) = \kappa^{-1}$ and $W(s) \sim (s-1) \kappa^{-1}$ for $s$ near 1. Since $1/(s-1)$ is not integrable near 1, the second integral explodes.

It follows that $\gamma$ is onto; the fact that $\gamma$ is increasing follows on differentiation. Indeed $\gamma'(v) = 1/((1-R)vW(v))$ and hence $h' = (1-R)hW(h)$. Also $h(-\infty):=\lim_{u \downarrow -\infty} h(u) =1$.

The first limit result for $\tilde{\gamma}$ follows immediately from the definition. For the second,
\begin{eqnarray*}
\lim_{v \uparrow \infty} e^{\gamma(v)}(1-W(v)) & = & \lim_{v \uparrow \infty} e^{-\tilde{\gamma}(v)}v^{1/(1-R)}(1-W(v)) = \lim_{v \uparrow \infty} e^{-\tilde{\gamma}(v)} \lim_{q \uparrow 1} N(q)^{1/(1-R)}(1-q) \\ & = & m(1)^{R/(1-R)} \lim_{q \uparrow 1} n(q)^{-R/(1-R)} = 1.
\end{eqnarray*}

\end{proof}

Define the candidate value function via
\begin{equation}
G(x,y,\theta,t) =
e^{-\beta t} \frac{x^{1 - R}}{1 - R} g \left(\frac{y
\theta}{x}\right), \hspace{20mm} x>0, \theta>0
\label{eq:vf}
\end{equation}
and extend the definition to $\theta=0$ and $-\theta y < x \leq 0$ by
\begin{eqnarray}
\label{eqn:G3defxleq0}
G(x,y,\theta,t) & = & e^{-\beta t}\frac{(x+y\theta)^{1-R}}{1-R}
\left(\frac{R}{\beta} \right)^R m(1)^{-R} \hspace{10mm} - \theta y < x \leq 0, \theta>0; \\
\label{eqn:G3deftheta=0}
G(x,y,0,t) & = & e^{-\beta t}\frac{x^{1-R}}{1-R} \left( \frac{R}{\beta} \right)^R \hspace{10mm} x \geq 0, \theta = 0.
\end{eqnarray}
Here continuity of $G$ at $x=0$ follows from the identity
\begin{equation}
\label{eqn:hlimit}
\lim_{z \uparrow \infty} z^{R-1} g(z) = \lim_{u \uparrow \infty} e^{-(1-R)u}h(u) = \lim_{v} e^{-(1-R)\gamma(v)}v = \lim_{v} e^{-(1-R)\tilde{\gamma}(v)} = m(1)^{-R}.
\end{equation}

\begin{lem}
\label{lem:2ndconcave}
Fix $y$ and $t$. Then $G=G(x,\theta)$ is concave in $x$ and
$\theta$ on $[0,\infty)\times[0,\infty)$. In particular, if $\psi(\chi) =
G(x - \chi y, y , \theta + \chi, t)$,
then $\psi$ is concave in $\chi$.
\end{lem}

\begin{proof}
The proof follows similarly to the proof of Lemma~\ref{lem:concave}, and makes use of the fact  $dh/du = (1-R)hW(h)$ proved in Lemma~\ref{lem:gamma}.
\end{proof}

\begin{lem} \label{lem:2ndoperator}
Consider the candidate function constructed in (\ref{eq:vf})--(\ref{eqn:G3deftheta=0}).
Then for $x>0,\theta>0$, $\mathcal{L} G = 0$, and $\mathcal{M} G \geq 0$. Further, $\mathcal{M} G = 0$ at $(x=0,\theta>0)$ and $\sL G = 0$ at $x=0$ and at $\theta=0$.
\end{lem}

\begin{proof}
The majority of the lemma follows exactly as in Lemma~\ref{lem:operator}.

For $\sM G|_{x=0}$, note that
$G_\theta |_{x=0} = y G (1-R)/(x+y\theta)|_{x=0} = (1-R)G/\theta$.
Then, $yG_x |_{x=0-} = y G (1-R)/(x+y\theta)|_{x=0-} = (1-R)G/\theta$, whereas for $x > 0$,
\[ yG_x = \frac{y(1-R)G}{x} - \frac{g'}{g} \frac{y^2 \theta}{x^2} G = \frac{(1-R)G}{\theta} \left[z -\frac{z^2 g'(z)}{(1-R)g(z)} \right], \]
and then for fixed $(y,\theta)$
\[ \lim_{x \downarrow 0} \left[ z -\frac{z^2 g'(z)}{(1-R)g(z)} \right] =\lim_{u \uparrow \infty} e^u \left( 1 - \frac{h'(u)}{(1-R)h(u)} \right) =  \lim_{v} e^{\gamma(v)} \left( 1 - W(v) \right) =1. \]
\end{proof}

\begin{proof}[Proof of Theorem \ref{thm:2ndmaincase}.]
For an admissible strategy $(C,\Theta) = (C_t, \Theta_t)_{t \geq 0}$
define the
process $M(C,\Theta) = (M_t)_{t
\geq 0}$ via
\begin{equation}
\label{eq:Mdef}
M_t = \int_0 ^t e^{-\beta s}
\frac{C_s ^{1-R}}{1-R}ds + G\left(X_t, Y_t, 0, t\right).
\end{equation}
where $G$ is as given in (\ref{eq:vf})---(\ref{eqn:G3deftheta=0}).

\noindent{\em Case 1: $\theta_0=0$ and $x_0>0$: we show $V =G$.}
For these initial values the agent
does not own any units of
asset for sale and consumption can only be financed from liquid (cash)
wealth. Then
$(\Theta_t)_{t\geq0} = 0$, $dX_t = -C_t dt$ and
the problem is non-stochastic. The candidate optimal consumption
function is $C(x,y,0) = \beta x/R$ and the associated consumption
process is
$C^* _t = \frac{\beta}{R}
x_0
e^{-\frac{\beta}{R}t}$ with resulting wealth process $X^*_t = x_0
e^{-\frac{\beta}{R} t}$.

Then the value function is
\begin{eqnarray*}
\mathbb{E} \left[\int_0^\infty e^{-\beta t}
\frac{{C^{*}_t}^{1-R}}{1-R} dt \right] & = &
\int_0^\infty e^{-\beta t} \left(\frac{\beta}{R}\right)^{1-R}
\frac{ \left(e^{-\frac{\beta}{R} t}
{x_0} \right)^{1-R}}{1-R} dt \\
& = & \left(\frac{R}{\beta}\right)^R \frac{x_0 ^{1 - R}}{1 - R}
= G(x_0,y_0,0,0),
\end{eqnarray*}
where the last equality follows from (\ref{eqn:G3deftheta=0}). Hence,
we have $V \geq G$.

Now consider general admissible strategies. Let $M^0$ be given by
$M^0_t = M_t(C_t,0)$.
Applying It\^{o}'s formula to $M^0$, we get
\begin{eqnarray*}
M^0_t - M^0_0 & = & \int_0^t \left[ e^{-\beta s}
\frac{C_s ^{1-R}}{1-R} - C_s G_x + \alpha Y_s G_y
+ \frac{1}{2} \eta^2 Y_s ^2 G_{yy} + G_s \right]ds \\
& & + \int_0^t \eta Y_s G_y dB_s \\
& = & N_t^1 + N_t^3.
\end{eqnarray*}
Lemma~\ref{lem:2ndoperator} implies that $\mathcal{L}G = 0$ and hence
$N_t^1 = 0$.

Suppose $R<1$. Then we have $0\leq M^0_t\leq M^0_0 + N^3_t$, and the
local
martingale $N^3_t$ is now bounded from below and hence a
supermartingale. Taking expectations we conclude $\mathbb{E} (M^0_t)
\leq
M^0_0 = G(x_0, y_0, 0, 0)$, and hence
\begin{equation}
G(x_0, y_0, 0, 0) \geq
\mathbb{E} \int_0^t e^{-\beta s} \frac{{C_s}^{1-R}}{1-R} ds
+ \mathbb{E} G(X_t, Y_t, 0, t) \geq
\mathbb{E} \int_0^t e^{-\beta s} \frac{{C_s}^{1-R}}{1-R} ds,
\label{eq:7}
\end{equation}
Letting $t \to \infty$, (\ref{eq:7}) we conclude
\[
G(x_0, y_0, 0, 0) \geq \mathbb{E} \int_0^\infty e^{-\beta t}
\frac{{C_t}^{1-R}}{1-R} dt.
\]
and taking a supremum over admissible strategies we have $G \geq V$, and
hence $G=V$.

For $R>1$, a modification of the proof of Theorem~\ref{thm:deg1}
applies
here also and $G=V$.

\noindent{Case 2: $x_0=0$ and $\theta_0>0$: we show $V  \geq G$.}
Under the candidate optimal strategy defined in Theorem~\ref{thm:2ndmaincase}
the consumption
and sale processes evolve according to $C_t dt = - Y_t d\Theta_t$,
meaning that the investor finances consumption only from the sales of
the
endowed asset and wealth stays constant and identically zero. In this
case, the proposed
strategies in (\ref{eq:optimal2}) become
\[
\Theta^*_t = \theta_0e^{-\frac{\beta}{R} \phi t},
\hspace{6mm}
C^*_t = \frac{\beta}{R} \phi Y_t \Theta^*_t  =
\frac{\beta}{R} \phi y_0\theta_0
\exp\left\{\beta (\epsilon - \delta^2/2 - \phi/R)t
+ \delta \sqrt{\beta}B_t  \right\}.
\]
where temporarily we write $\phi = m(1)= \delta^2 R(1-R)/2 - \epsilon(1-R)
+ 1>0$.

The corresponding value function is
\begin{eqnarray*}
G^* & = & \mathbb{E} \left[ \int_0^ \infty e^{-\beta t}
\frac{{C^*_t}^{1-R}}{1-R} dt \right] \\
& = &\left(\frac{\beta}{R}\right)^{1-R} \frac{(\phi y_0 \theta_0)^{1-R}}{1-R}
\mathbb{E}\left[ \int_0^\infty
e^{-\beta t}
e^{ (1-R) \beta ( \epsilon - \frac{\delta^2}{2} - \frac{\phi}{R} )t +
 \delta \sqrt{\beta}(1-R) B_t }
dt \right] \\
& = & \left(\frac{\beta}{R}\right)^{1-R}
\frac{(\phi y_0 \theta_0)^{1-R}}{1-R}
\int_0^\infty
e^{ \{
(\epsilon (1 - R) - \frac{\delta^2}{2} R(1 - R) - 1)
- \frac{(1-R)}{R} \phi  \} \beta t}  dt \\
& = &  \left(\frac{R}{\beta}\right)^{1-R}
\frac{(\phi y_0\theta_0)^{1-R}}{1-R} \int_0^\infty e^{-(\beta \phi/R)t}
dt =
 \left(\frac{R}{\beta}\right)^R \frac{(y_0 \theta_0)^{1-R}}{1-R}
\phi^{-R} =  G(0,y_0,\theta_0,0).
\end{eqnarray*}

Then, under the candidate optimal strategy,
\[
G(0,y_0,\theta_0,0) =  \mathbb{E} \left[ \int_0^ \infty e^{-\beta t}
\frac{(C^*_t)^{1-R}}{1-R} dt \right],
\]
and we have $G(0,y_0,\theta_0,0) \leq V(0,y_0,\theta_0,0)$.

\noindent{Case 3: $x_0>0$ and $\theta_0>0$: we show $V  \geq G$.}
Let $M^* = M(C^*,\Theta^*)$ for the candidate optimal strategies in
Theorem~\ref{thm:2ndmaincase}.

From the form of the candidate value function we know that $G$
is $C^{1,2,1,1}$.
Then applying It\^{o}'s formula to $M^*$, we have
\begin{eqnarray}
M^*_{t}-M^*_{0} & = & \int_{0}^{t}\left[e^{-\beta
s}\frac{(C^*_s)^{1-R}}{1-R}-C^*_sG_{x}+\alpha
Y_sG_{y}+\frac{1}{2}\eta^{2}Y_s^{2}G_{yy}+G_{t}\right]ds\nonumber \\
&& + \int_{(0,t]}\left(G_{\theta}-Y_sG_{x}\right)d\Theta_s \label{eq:9} \\
&& +  \int_{0}^{t}\eta Y_s G_{y}dB_{s}\nonumber \\
& =: & N_{t}^{1}+N_{t}^{2}+N_{t}^{3}. \nonumber
\end{eqnarray}
Since $C^*_s = G_x^{-1/R} e^{\beta s/R}$ is optimal and, by
Lemma~\ref{lem:2ndoperator}, $\mathcal{L}G = 0$, we have $N^1_t = 0$.
Further, under the proposed strategies in (\ref{eq:optimal2}),
$d\Theta_t \neq 0$ if and only if
$X_t = 0$.
Then,
by Lemma~\ref{lem:2ndoperator}, $\sM G|_{x=0} = 0$ and $N^2_t = 0$.

The following Lemma is proved in the appendix.
\begin{lem}
\label{lem:mgc3}
(1) $N^3$ given by $N_{t}^{3}=\int_{0}^{t}\eta Y_s G_{y}(X^*_s, Y_s,
\Theta^*_s,s) dB_{s}$ is a martingale. \\
(2) $\lim_{t \uparrow \infty} \mathbb{E} [G(X^*_t, Y_t, \Theta^*_t,t)] =
0$
\end{lem}

The conclusion that $V \geq G$ now follows exactly as in the proof of Theorem~\ref{thm:maincase} but using Lemma~\ref{lem:mgc3} in place of Lemma~\ref{lem:mg}.

\noindent{Case 4: $x_0 \geq 0$ and $\theta_0>0$: $V  \leq G$.}
To complete the proof of the theorem, it remains to show for
$\theta_0 > 0$ and
general admissible strategies, we have
$V(x_0, y_0, \theta_0, 0) \leq G(x_0, y_0, \theta_0, 0)$.
Recall the definition of $M$ in (\ref{eq:Mdef}).

Applying the generalised It\^{o}'s
formula~\cite[Section 4.7]{gito} to $M_{t}$ leads to the expression in (\ref{eqn:Msde}) and
\[ M_{t}-M_{0}
  =  N_{t}^{1}+N_{t}^{2}+N_{t}^{3}+N_{t}^{4}.  
  \]
Lemma~\ref{lem:2ndoperator} implies that under general admissible
strategies, $N^1 _t \leq 0$, and $N^2 _t \leq 0$ with equality at $x=0$.
Consider the jump term,
\begin{equation}
N_{t}^{3}= \sum_{0 \leq s\leq
t}\left[G\left(X_{s},Y_{s},\Theta_{s},s\right)
-G\left(X_{s-},Y_{s},\Theta_{s-},s\right) - G_x (\Delta X)_s
 - G_\theta (\Delta \Theta)_s
\right]
\end{equation}
Using the fact that $(\Delta X)_s = -Y_s (\Delta \Theta)_s$ and
writing $\theta = \Theta_{s-}$, $x = X_{s-}$, $\chi = - (\Delta
\Theta)_s$ each non-zero jump in $N^3$ is of the form
\[
(\Delta N^3)_s = G(x+y \chi,y,\theta - \chi,s) - G(x,y,\theta,s) +
\chi \left[ G_\theta (x,y,\theta,s) - y G_x (x,y,\theta,s) \right].
\]
Note that by Lemma~\ref{lem:2ndconcave}, $G(x+y \chi,y,\theta -
\chi,s)$ is
concave in $\chi$ and hence
$(\Delta N^3) \leq 0$.

For the case $R<1$ the remainder of the proof follows as in the proof
of Theorem~\ref{thm:deg1}.
The case $R>1$ for general admissible strategies is covered in
Appendix~\ref{app:R>1}.

\end{proof}

\input{SpComparativeStats140815}

\appendix

\input{appendix140815}

\end{document}

%% file: SpComparativeStats140815.tex
\section{Comparative statics}

In this section, we provide comparative statics describing how the outputs of the model depend on market parameters. This section consists of five parts, analysis of the optimal threshold $z^*$, the value function $g$, the optimal consumption $C(x,y,\theta)$, the utility indifference price $p(x,y,\theta)$, and the cost of illiquidity $p^*(x,y,\theta)$, and are based on our numerical results. The cost of illiquidity, defined in (\ref{eq:defill}) below represents the loss in cash terms faced by our agent when compared with an otherwise identical agent with the same initial portfolio who is able to adjust her portfolio of the risky asset in either direction at zero cost.

The equations describing the function $n$ and the first crossing of $m$ are simple to implement in MATLAB, and then it also proved straightforward to calculate $h$ or $\gamma$ and thence the value function in the non-degenerate cases. Figures~\ref{fig:mnR<1} and \ref{fig:mnR>1} are generic plots of the various functions used in the construction of the value function. The parameter values are such that we are in the second non-degenerate case ($\epsilon \geq \delta^2 R$ and $\epsilon < \frac{\delta^2 R}{2} + \frac{1}{1-R}$ if $R<1$), but the figures would be similar for the first non-degenerate case ($0<\epsilon < \delta^2 R$ and $\epsilon < \frac{\delta^2 R}{2} + \frac{1}{1-R}$ if $R<1$). The two figures cover the cases $R<1$ and $R>1$ respectively.
For $R<1$, as plotted in Figure~\ref{fig:mnR<1}, $m$ and $n$ are monotone decreasing and $W$ is increasing on $[1,\infty)$ with $\lim_{v \to 1}W(v) = 0$ and $\lim_{v \to \infty} W(v) = 1$. Further, we have $\gamma(v)$ is increasing on $[1,\infty)$ and $g$ is concave and increasing. For $R>1$, as plotted in Figure~\ref{fig:mnR>1}, $m$ and $n$ are monotone increasing and $W$ is decreasing on $(0, 1]$ with $\lim_{v \to 0}W(v) = 1$ and $\lim_{v \to 1} W(v) = 0$. Finally, we have $\gamma(v)$ is decreasing on $(0,1]$ and $g$ is convex decreasing and convergent to zero as $z$ tends to infinity.

\begin{figure}
\begin{centering}
\includegraphics[scale=0.295]{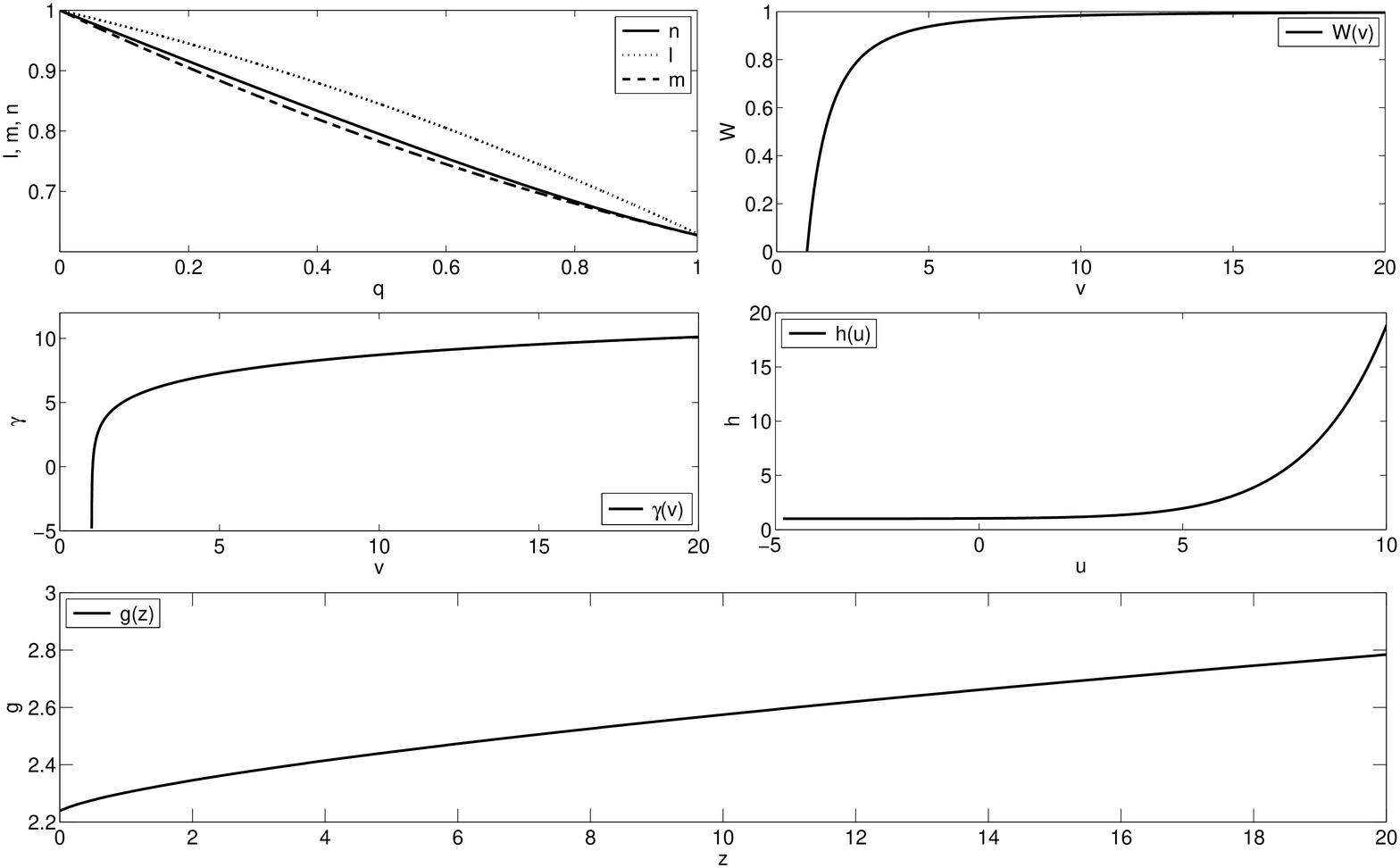}
\par\end{centering}

\caption{Transformations from $m, n, \ell$ to $W(v)$ to $\gamma(v)$ to $h(u)$ and $g(z)$ in the second non-degenerate scenario in the case $R<1$. Parameters are $\epsilon = 1$ $\delta = 1$, $\beta = 0.1$ and $R = 0.5$. For these parameters $m$ is monotonic decreasing.}
\label{fig:mnR<1}
\end{figure}

\begin{figure}
\begin{centering}
\includegraphics[scale=0.29]{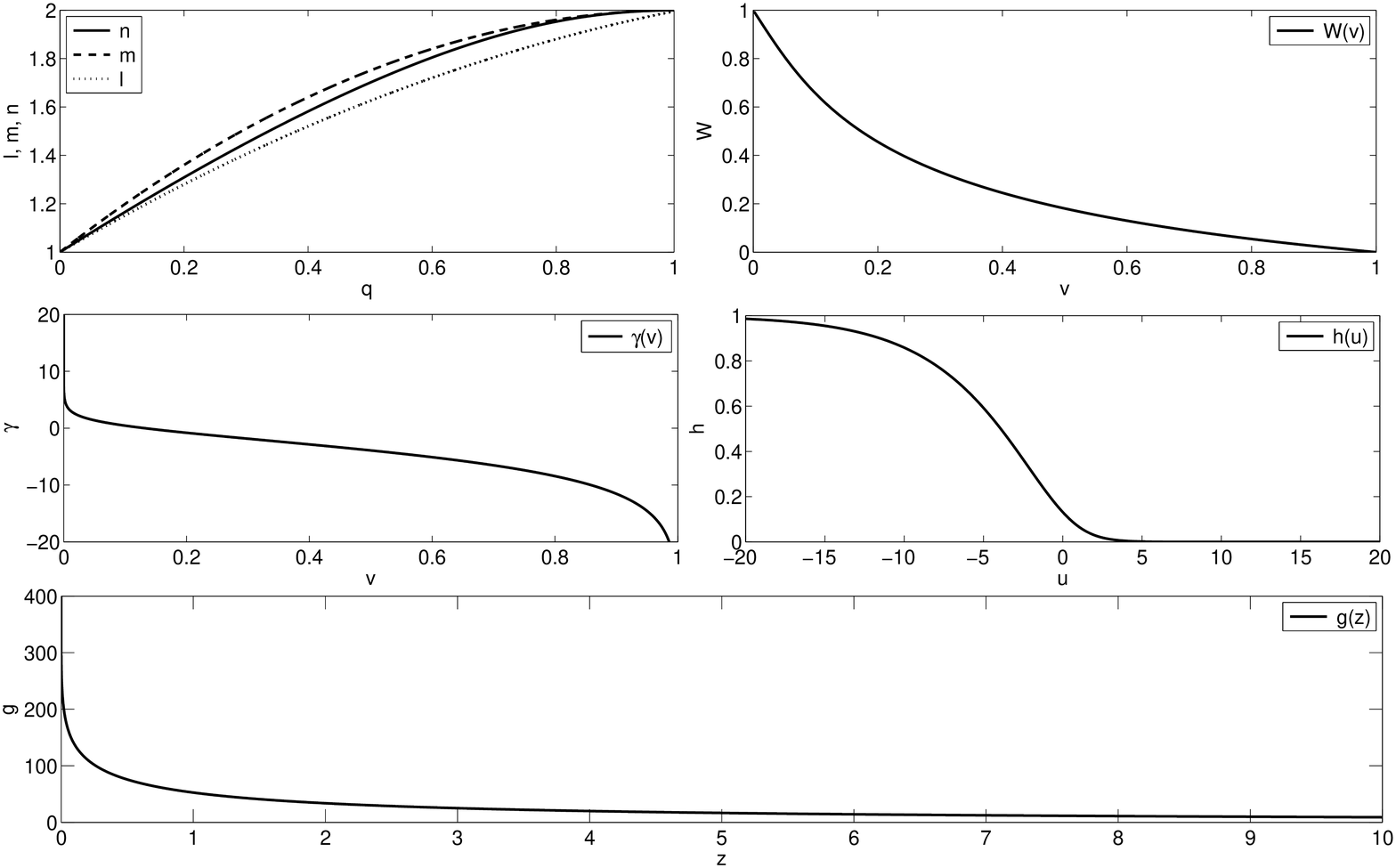}
\par\end{centering}

\caption{Transformations from $m, n, \ell$ to $W(v)$ to $\gamma(v)$ to $h(u)$ and $g(z)$ in the second non-degenerate scenario in the case $R>1$. Parameters are $\epsilon = 3$ $\delta = 1$, $\beta = 0.1$ and $R = 2$.}
\label{fig:mnR>1}
\end{figure}

Figures~\ref{fig:z*1} and \ref{fig:z*2} show that $z^*$ increases as mean return $\epsilon$ increases and decreases as volatility $\delta$ increases or risk aversion $R$ increases. As $\epsilon$ increases, the non-traded asset $Y$ becomes more valuable and it is optimal for the investor to wait longer to sell $Y$ for a higher return. For $\epsilon = 0$, when the endowed asset has zero return but with additional risk, the optimal strategy is to sell immediately to remove the risk. Similarly, as $\delta$ increases, the level of $z^*$ decreases as holding $Y$ involves additional risk. Hence, it is optimal for the investor to sell units of $Y$ sooner in order to mitigate this risk. As the risk aversion of the investor increases, she is less tolerant to the risk of the endowed asset and hence more inclined to sell $Y$ earlier. As $R \to 0$, (provided $\epsilon > 0$) we have $z^* \to \infty$, which implies the optimal strategy is never to sell the asset. In the limit the investor is not concerned about the risk of holding the risky asset. Conversely, as $R \to \infty$, we have $z^* \to 0$. In this case, the investor cannot tolerate any risks and it is therefore optimal to sell the asset immediately to arrive at a safe position.

\begin{figure}
\begin{centering}
\includegraphics[scale=0.3]{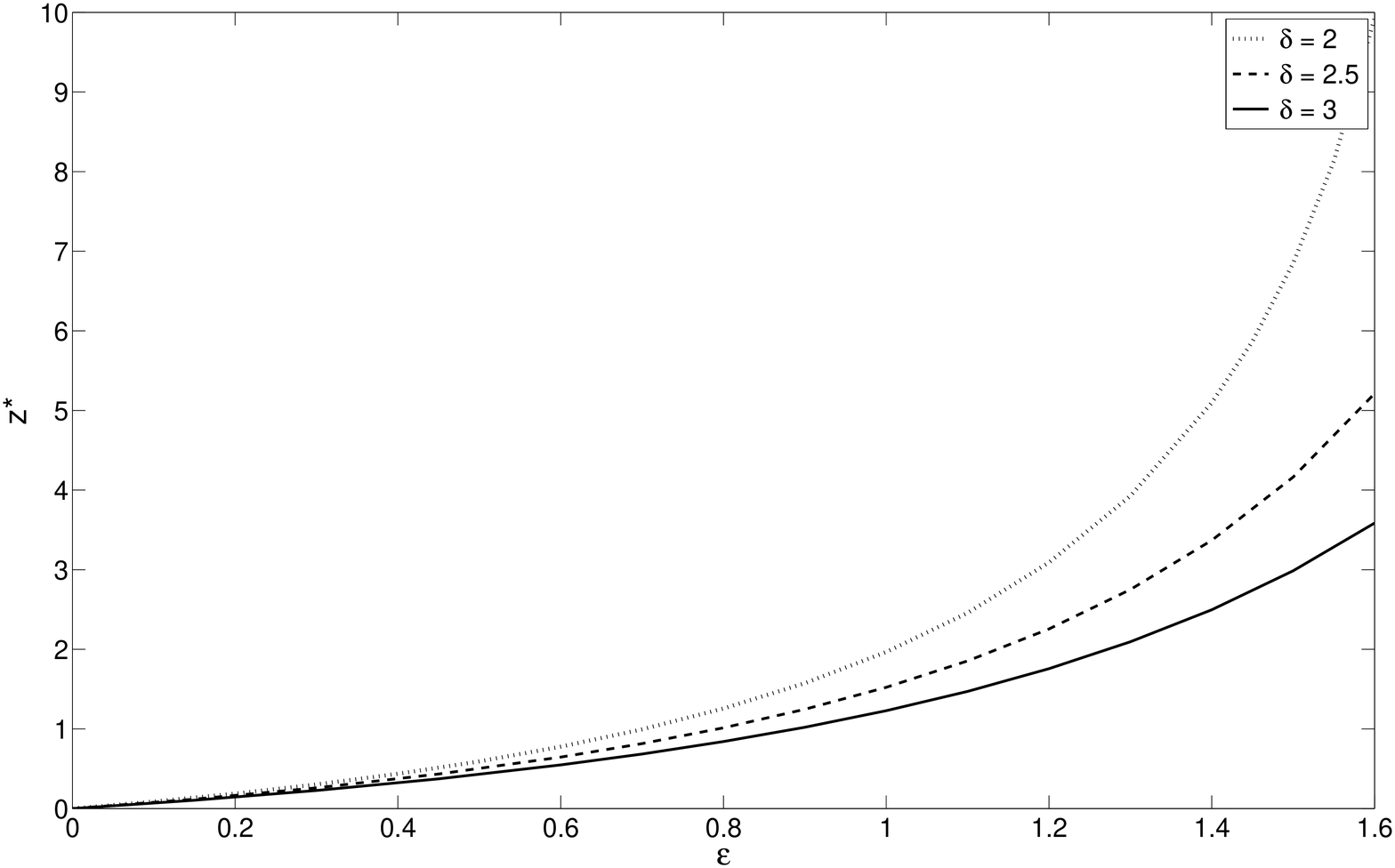}
\par\end{centering}

\caption{$z^{*}$ increases as $\epsilon$ increases or as $\delta$ increases. Here $\beta = 0.1$ and $R = 0.5$.}
\label{fig:z*1}
\end{figure}

\begin{figure}
\begin{centering}
\includegraphics[scale=0.32]{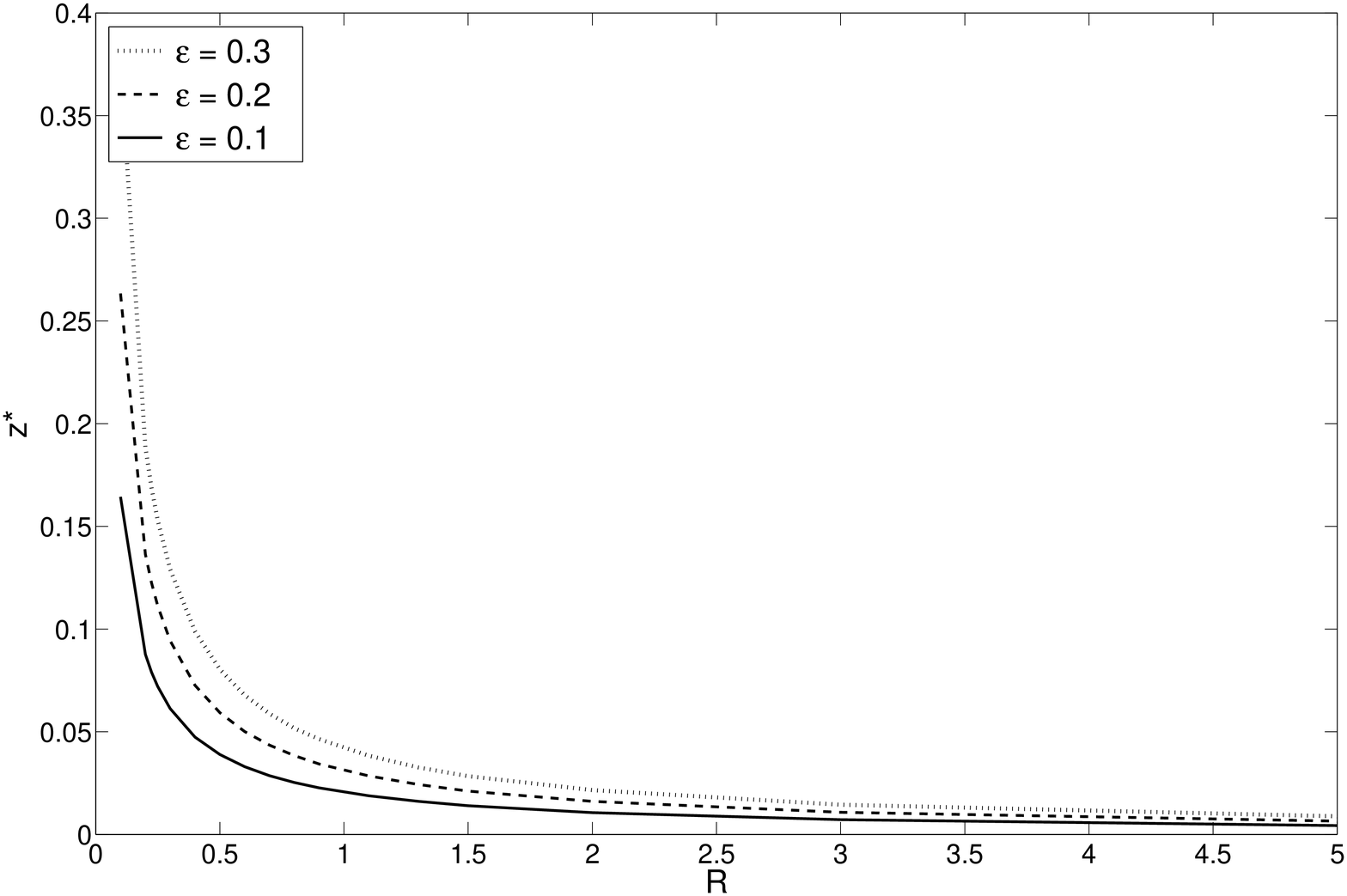}
\par\end{centering}

\caption{$z^{*}$ decreases as $R$ increases or as $\epsilon$ decreases. Here $\delta = 3$ and $\beta = 0.1$.}
\label{fig:z*2}
\end{figure}

The value function as expressed via $g$ in non-degenerate cases is plotted in Figures~\ref{fig:gep} and \ref{fig:gR} under different drifts and risk aversions. These figures show that $g$ is increasing in drift while $g$ has no monotonicity in risk aversion. (A similar plot shows that $g$ is decreasing in volatility.) As the non-traded asset becomes more valuable, the investor can choose optimal sale and consumption strategies which lead to a larger value function. (Further, as the asset becomes more risky, the additional risk makes the value function smaller.) Meanwhile, as $\epsilon$ increases, $z^*$ in Figure~\ref{fig:gep} is decreasing (and as $\delta$ increases, $z^*$ is increasing). These results are consistent with the results in described in the previous paragraph. At $z = z^*$, smooth fit conditions are satisfied. Observe also that for different values of drift, we nonetheless have that $g$ starts at the same point. This corresponds to the value function when $\theta_0 = 0$ whereby consumption is only financed by initial wealth and the problem is deterministic. In this case, we have $g(0) = (R /\beta)^R$.

\begin{figure}
\begin{centering}
\includegraphics[scale=0.29]{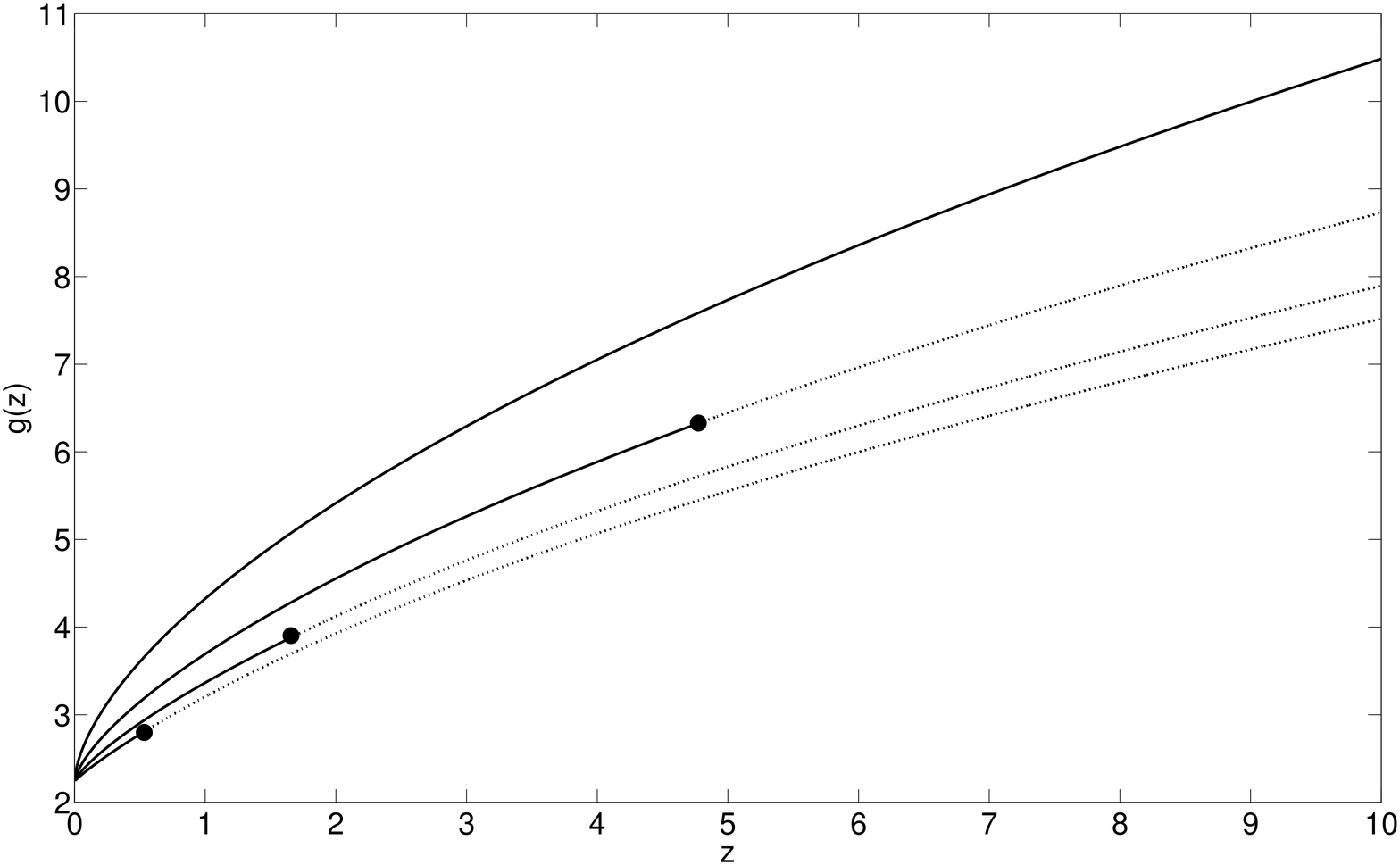}
\par\end{centering}

\caption{$g(z)$ with different $\epsilon$ in the first and second non-degenerate scenarios. Dotted line: $z \geq z^{*}$, solid line: $z \leq z^{*}$ and dots represent $z^{*}$. $\epsilon$ varies from top to bottom as 2, 1.5, 1, 0.5, with fixed parameters $\delta = 2$, $\beta = 0.1$ and $R = 0.5$. The top line is the value function $g$ in the second non-degenerate scenario given $\epsilon = \delta^2 R = 2$ and $z^*$ is infinite.}
\label{fig:gep}
\end{figure}



\begin{figure}
\begin{centering}
\includegraphics[scale=0.27]{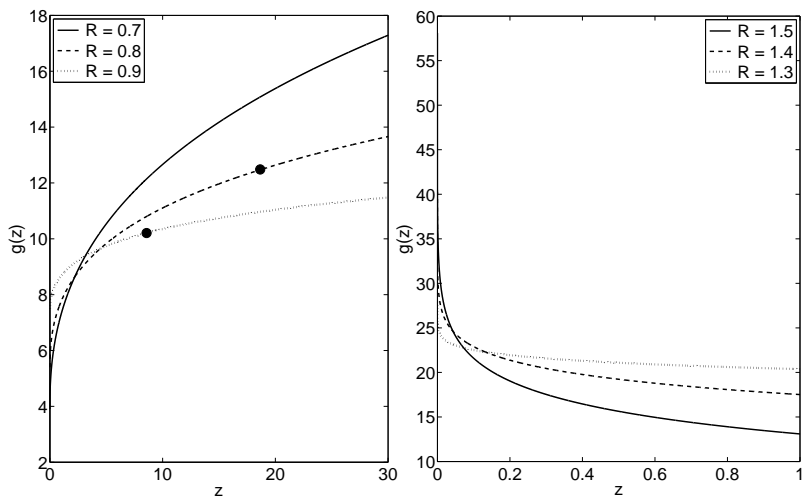}
\par\end{centering}
\caption{$g(z)$ with different risk aversion $R$ in the first and second non-degenerate scenarios. In the left graph, $R$ takes values in 0.7, 0.8 and 0.9. The rest of the parameters are $\epsilon = 3$, $\delta = 2$, $\beta = 0.1$. The critical risk aversion is $R = \epsilon/ \delta^2 = 0.75$. The dots represent finite $z^*$ and the solid line is the value function $g$ in the second non-degenerate scenario with infinite $z^*$. In the right graph, $R$ takes values in 1.3, 1.4 and 1.5 and the rest of the parameters are $\epsilon = 6$, $\delta = 2$ and $\beta = 0.1$.}
\label{fig:gR}
\end{figure}

Optimal consumption $C(x,y,\theta)$ is considered in Figures~\ref{fig:CthetaR}---\ref{fig:Cthetaep}. Figure~\ref{fig:CthetaR} plots the optimal consumption $C(1,1,\theta)$ as a function of endowed units $\theta$ and shows that the optimal consumption increases in $\theta$: as the size of the holdings of the non-traded asset $Y$ increases, the agent feels richer
and hence consumes at a faster rate. For $\theta = 0$, the optimal consumption
$C(x,y,0) = x g(0)^{-\frac{1}{R}}= \frac{\beta}{R} x$ is strictly positive and is financed from cash wealth. Figure~\ref{fig:CthetaR} also suggests that the optimal consumption $C(1,1,\theta)$ decreases in risk aversion. Given the set of parameters the critical risk aversion (i.e. the boundary between the two non-degenerate cases) is at $R = \epsilon/\delta^2 = 0.75$. For  the bottom two lines in Figure~\ref{fig:CthetaR} with $R>0.75$, we have $\epsilon<\delta^2 R$ and this falls into the first non-degenerate case with finite $z^*$. For $R\leq 0.75$, we have $\epsilon \geq \delta^2 R$, which is the second non-degenerate case with infinite $z^*$. As we see, there is no discontinuity in consumption with respect to risk aversion at either $R = 0.75$ or $R = 1$.
The optimal consumptions for different risk aversions differ primarily in the levels, and the dominant factor is the optimal consumption for $\theta = 0$. As argued above $C(x, y, 0) = \beta x/R$ is decreasing in $R$.

\begin{figure}
\begin{centering}
\includegraphics[scale=0.27]{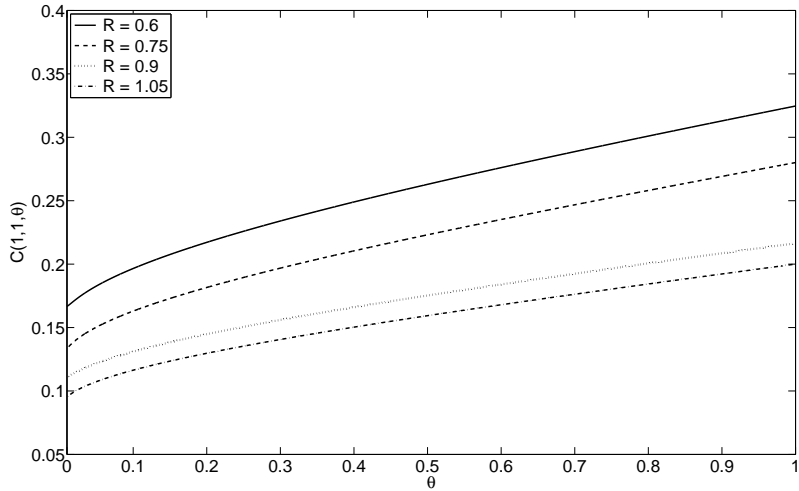}
\par\end{centering}

\caption{Optimal consumption $C(1,1,\theta)$ as $R$ varies. $R$ takes values in 0.6, 0.75, 0.9, 1.05 with parameters $\epsilon = 3$, $\delta = 2$, $\beta = 0.1$ and $\theta \in [0,1]$. The critical risk aversion is $R = \epsilon/\delta^2 = 0.75$. The top two lines correspond to the optimal consumption in the second non-degenerate scenario where $z^*$ is infinite under the condition that $\epsilon \geq \delta^2 R$. The bottom two lines correspond to the first non-degenerate case with finite $z^*$.}
\label{fig:CthetaR}
\end{figure}

Figure~\ref{fig:CxR} plots both consumption as a function of wealth $C(x, 1, 1)$ and the ratio of consumption to wealth $C(x,1,1)/x$ as a function of $x$ with different risk aversions. Note that this can only be shown for
$x > y\theta/z^* = 1/z^*$ since if $x < 1/z^*$ the agent makes an immediate sale of units of risky asset. The critical value of the risk aversion is $R = \epsilon/\delta^2 = 0.75$. For $R>0.75$, we have $z^*<\infty$ and $x^* = 1/z^*>0$ while for $R \leq 0.75$, $z^* = \infty$ and $x^* = 1/z^* = 0$. The results show that the optimal rate of consumption is an increasing function of wealth but that consumption per unit wealth is a decreasing function of wealth. (In the standard Merton problem, consumption is proportional to wealth.) As the agent becomes richer, she consumes more, but the fraction of wealth that she consumes becomes smaller. The explanation is that her endowed wealth is being held constant. By scaling we have that if both $x$ and $\theta$ are increased by the same factor, then consumption would also rise by the same factor, but here $x$ is increasing, but $\theta$ (and $y$) are held constant, and hence consumption increases more slowly than wealth. In the limit $x \to \infty$ we have $\lim_{x \to \infty} C(x,1,1) = \infty$ and $\lim_{x \to \infty} C(x,y,\theta)/x = g(0)^{-\frac{1}{R}} = \beta/R$.

\begin{figure}
\begin{centering}
\includegraphics[scale=0.27]{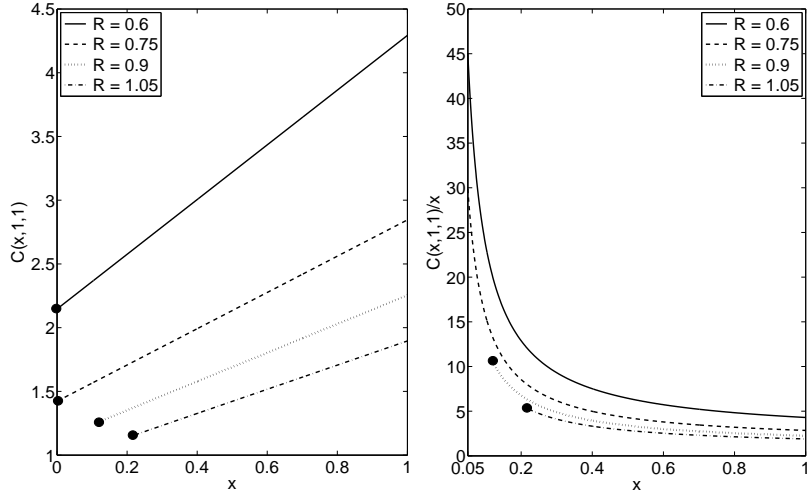}
\par\end{centering}

\caption{Optimal consumption $C(x,1,1)$ and $C(x,1,1)/x$ as $R$ varies. $R$ takes values in 0.6, 0.75, 0.9 and 1.05 with parameters $\epsilon = 3$, $\delta = 2$, $y_0 = 1$ and $\theta_0 = 1$. The dots represent $x^* = 1/z^*$ and the critical risk aversion is $R = \epsilon/\delta^2 = 0.75$. In both graphs, the top two lines correspond to the optimal consumptions in the second non-degenerate case with $x^* = 0$. The bottom two lines are the optimal consumptions in the first non-degenerate case with finite $z^*$, or equivalently, $x^*>0$. }
\label{fig:CxR}
\end{figure}

Figure~\ref{fig:Cthetaep} plots the optimal consumption $C(1,1,\theta)$ as a function of $\theta$ and $\epsilon$.
 Here we find a first surprising result: we might expect the optimal consumption $C(x,y,\theta)$ to be increasing in the drift, but this is not the case for large $\theta$. For an explanation of this phenomena, recall that the optimal exercise ratio $z^*$ is increasing in the drift. As the drift increases, the asset has a more promising return on average which makes the agent feel richer and consume at a higher rate. However, a larger drift also implies a larger $z^*$, indicating that the agent should postpone the sale of the risky asset. Hence, a larger drift involves more risk, and in order to mitigate this risk, the agent consumes less in the short term. Hence, the optimal consumption decreases in the drift for large $\theta$.
We find similar results if we consider $C(1,1,\theta)$ as a function of $\delta$.
 Optimal consumption is not necessarily decreasing in volatility and consumption can be increasing in volatility for large values of $\theta$.
 Analogously, if we plot $C(x,1,1)$ we find that consumption is a decreasing (increasing) function of return $\epsilon$ if wealth $x$ is small (large).

\begin{figure}
\begin{centering}
\includegraphics[scale=0.3]{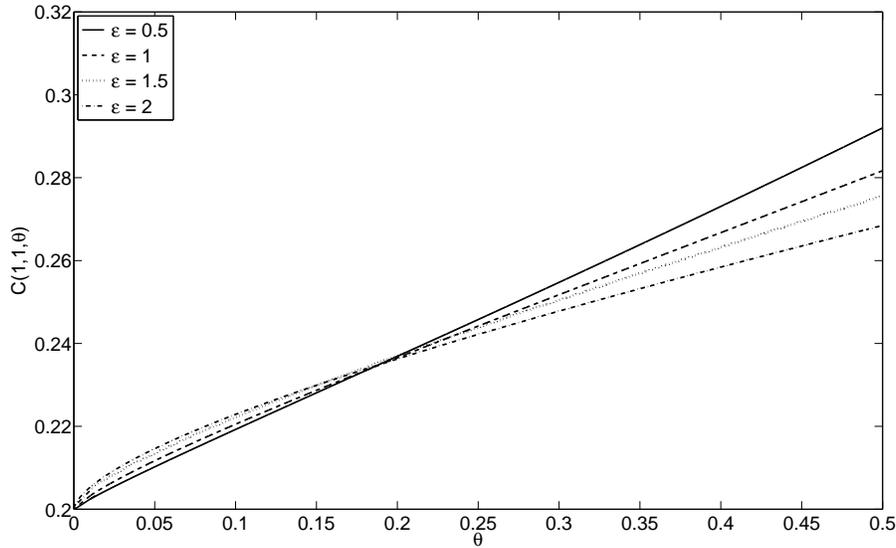}
\par\end{centering}

\caption{Optimal consumption $C(1,1,\theta)$ as $\epsilon$ varies. $\epsilon$ takes values in 0.5, 1, 1.5 and 2 with parameters $\delta = 2$, $\beta = 0.1$, $R = 0.5$, $x_0 = 1$ and $y_0 = 1$. The critical mean return is $\epsilon = \delta^2 R = 2$. When $\epsilon = 2$ we are in the second non-degenerate case.}
\label{fig:Cthetaep}
\end{figure}

Figures~\ref{fig:pxep}---\ref{fig:pxR} plot the utility indifference price or certainty equivalence value $p(x, y, \theta)$. Recall that in the second and third cases of Theorem~\ref{thm:4cases} the certainty equivalent value of the non-traded asset is given by
\[
p(x,y,\theta) = x \left[\frac{g\left(\frac{y\theta}{x}\right)}{g(0)}\right]^{\frac{1}{1-R}} - x
\]
Figures~\ref{fig:pxep} and \ref{fig:pxvol} consider the indifference price as a function of wealth. Dots in figures represent the optimal exercise ratio $z^* = y\theta/x$.
In each of the figures we choose a range of parameter values such that sometimes we are in the first non-degenerate case, and sometimes in the second non-degenerate case. In Figure~\ref{fig:pxep}, for $\epsilon<2$, we have $z^*<\infty$ and $x^* = 1/z^*>0$, and for $\epsilon \geq 2$, we have $z^* = \infty$ and $x^* = 0$.
We can see $p(x, 1, 1)$ is concave and increasing in $x$. It follows from
Theorem~\ref{thm:maincase} that $g(z) = (R/\beta)^R m(q^*)^{-R} (1 + z)^{1-R}$ for $z \geq z^*$. Further, under the condition that $0<\epsilon<\delta^2 R$ and $\epsilon<\frac{\delta^2}{2}R + \frac{1}{1-R}$, which ensures a finite exercise ratio,
\[
\lim_{x \to 0}p(x,y,\theta) = \lim_{x \to 0} x\left\{ \left[\frac{g\left(\frac{y\theta}{x}\right)}{g(0)}\right]^{\frac{1}{1-R}} - 1\right\} = \lim_{x \to 0}\left\{m(q^*)^{\frac{R}{R - 1}} (x + y\theta) - x\right\} = m(q^*)^{\frac{R}{R-1}} y\theta > y\theta.
\]
In that case, for $x = 0$, where no initial wealth is available to finance consumption, it is optimal for the investor to sell some units of the endowed asset $Y$ immediately so as to keep the ratio of the wealth invested in the endowed asset to liquid wealth below $z^*$, i.e. from the initial portfolio ($x = 0$, $\theta = \Theta_{0-}$)
the agent moves to ($x = X_{0+}$, $\theta = \Theta_{0+}$), where
$\Theta_{0+} = \frac{z^*}{1 + z^*} \Theta_{0-}$ and $X_{0+} = \frac{1}{1 + z^*}y\Theta_{0-}$. The monotonicity of $p(x, 1, 1)$ in $\epsilon$ and $\delta$ is also illustrated in Figures~\ref{fig:pxep} and \ref{fig:pxvol}: a higher mean return adds value to the asset, while the increasing volatility makes $Y$ more risky and reduces value. Also observe that for the drift larger than the critical value, the change in drift does not move the dot (representing the critical ratio) while for the drift smaller than the critical value, the dot  moves rightwards as drift increases. To the left of the dot, the agent should sell the endowed asset initially, while to the right of the dot, the agent should wait. As drift increases, the agent should wait longer for a higher return when selling the asset.

\begin{figure}
\begin{centering}
\includegraphics[scale=0.3]{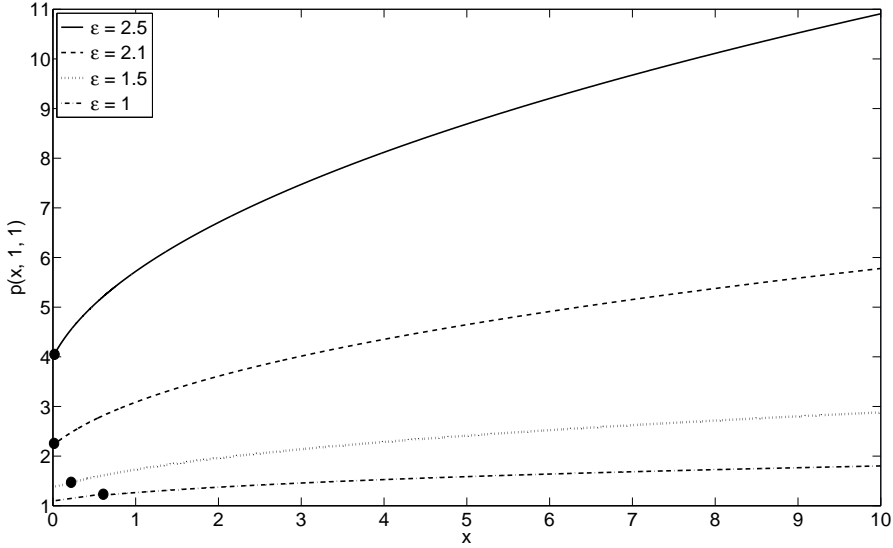}
\par\end{centering}

\caption{Indifference price $p(x,1,1)$ as $\epsilon$ varies. $\epsilon$ varies from top to bottom as 2.5, 2.1, 1.5, 1 with fixed parameters $\delta = 2$, $\beta = 0.1$, $R = 0.5$, $\theta_0 = 1$ and $y_0 = 1$. The dots represent $x^* = 1/z^*$ and the critical mean return is $\epsilon = \delta^2 R = 2$. }
\label{fig:pxep}
\end{figure}

\begin{figure}
\begin{centering}
\includegraphics[scale=0.28]{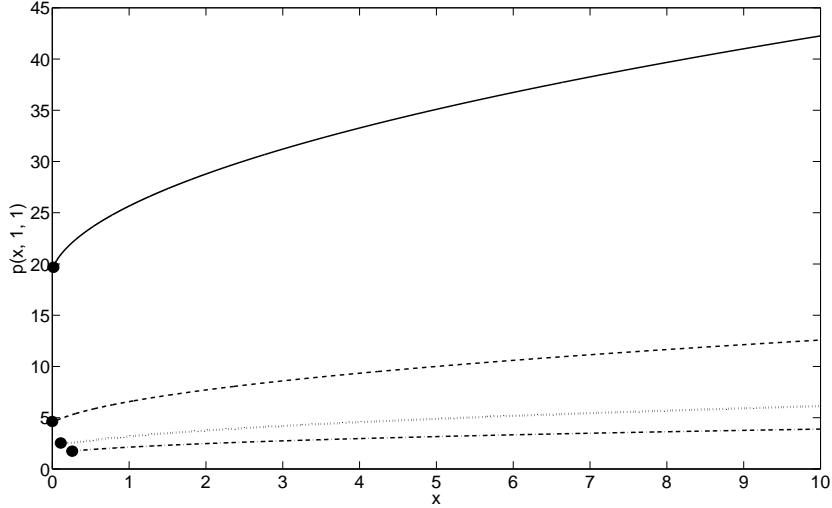}
\par\end{centering}

\caption{Indifference price $p(x,1,1)$. $\delta$ varies from top to bottom as 2.1, 2.4, 2.8 and 3.2 with fixed parameters $\epsilon = 3$, $\beta = 0.1$, $R = 0.5$, $\theta_0 = 1$ and $y_0 = 1$. The dots represent $x^* = 1/z^*$ and the critical volatility is $\delta = \sqrt{\epsilon/R} = 2.45$. The top two lines correspond to the indifference prices in the second non-degenerate case with $x^* = 0$. The bottom two lines are indifference prices in the first non-degenerate case with $x^*>0$.}
\label{fig:pxvol}
\end{figure}

Figure~\ref{fig:pthetaep} considers the indifference price $p(1,1,\theta)$ and unit indifference price $p(1,1,\theta)/\theta$ as a function of $\theta$. We see that $p(1, 1, \theta)$ is increasing in $\theta$ and for $\theta$ = 0, $p(1, 1, 0) = 0$, reflecting the fact that a null holding is worth nothing. We also have the unit price $p(1, 1, \theta)/\theta$ is decreasing in the units of asset $\theta$. For small holdings, the marginal price $\lim_{\theta \to 0} p(1,1,\theta)/\theta$ is infinite. As $\theta \to \infty$, the figures imply that the unit price $p(1, 1, \theta)/\theta$ tends to some constant larger than the unit price $y$ of $Y$:
\[
\lim_{\theta \to \infty} \frac{p(x,y,\theta)}{\theta} =
\lim_{\theta \to \infty} \frac{x \left[\frac{g\left(\frac{y\theta}{x}\right)}{g(0)}\right]^{\frac{1}{1-R}} - x}{\theta}
= \lim_{\theta \to \infty} \frac{m(q^*)^{\frac{R}{R - 1}} (x + y\theta) - x}{\theta} = m(q^*)^{\frac{R}{R - 1}}y > y,
\]
where the second equality follows since for $z \geq z^*$, we have $g(z) = (R/\beta)^R m(q^*)^{-R} (1+z)^{1-R}$.

\begin{figure}
\begin{centering}
\includegraphics[scale=0.28]{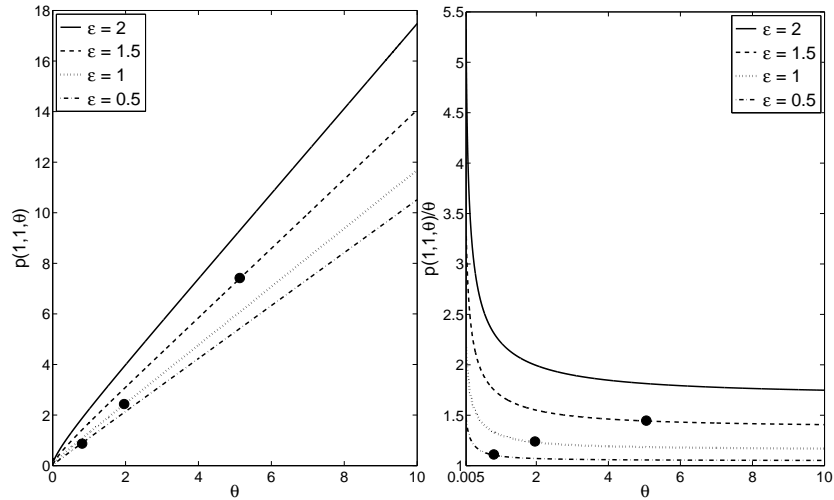}
\par\end{centering}

\caption{Indifference price $p(1,1,\theta)$ and unit price $p(1,1,\theta)/\theta$. $\epsilon$ varies from top to bottom as 2, 1.5, 1, 0.5 with fixed parameters $\delta = 2$, $\beta = 0.1$, $R = 0.5$, $x_0 = 1$ and $y_0 = 1$. The dots represent $\theta^* = z^*$ and the critical mean return is $\epsilon = \delta^2 R = 2$. The top line corresponds to the indifference price in the second non-degenerate case with infinite $z^*$.}
\label{fig:pthetaep}
\end{figure}

Figure~\ref{fig:pthetaep} also illustrates the monotonicity of $p$ in the drift parameter $\epsilon$ and we have $p(1,1,\theta)$ and $p(1,1,\theta)/\theta$ both increase in the drift.
Similarly, it can be shown that $p(1,1,\theta)$ and $p(1,1,\theta)/\theta$ are both decreasing in $\delta$, reflecting the increased riskiness of positions as volatility increases.

Figure~\ref{fig:pxR} plots the indifference price as a function of cash wealth for different risk aversions.
Naively we might expect the price to be monotone decreasing in risk aversion - a more risk averse agent will assign a lower value to a risky asset. However, the results show that this not the case, and for large wealths the utility indifference price is increasing in $R$. (If we fix wealth $x$ and consider the certainty equivalent value as a function of quantity $\theta$ then we find a similar reversal, and the certainty equivalent value is increasing in $R$ for small $\theta$.)

\begin{figure}
\begin{centering}
\includegraphics[scale=0.29]{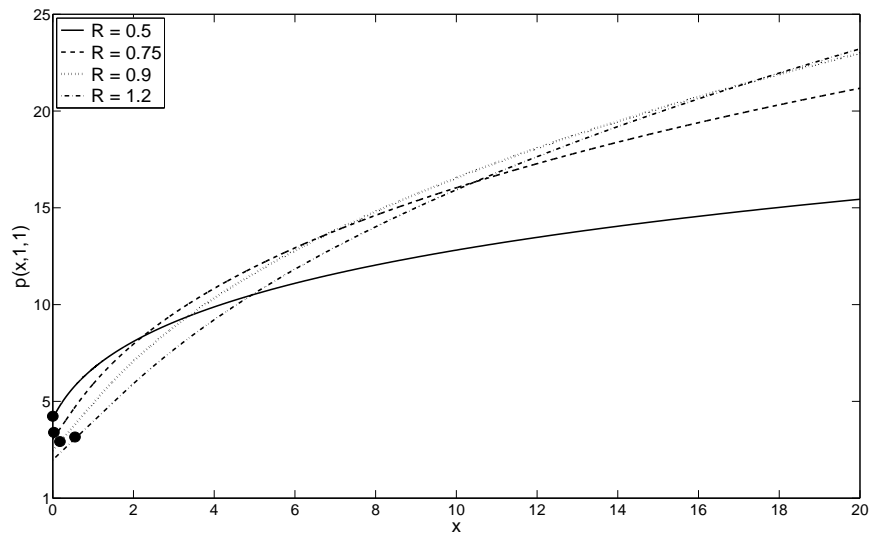}
\par\end{centering}

\caption{Indifference price $p(x,1,1)$. $R$ takes values in 0.5, 0.75, 0.9 and 1.2 with fixed parameters $\epsilon = 3$, $\delta = 2$, $\beta = 0.1$, $y_0 = 1$ and $\theta_0 = 1$. The dots represent $x^* = 1/z^*$ and the critical risk aversion is $R = \epsilon/ \delta^2 = 0.75$. The top two lines for $x\in[0,1]$ correspond to the indifference prices in the second non-degenerate case with $x^* = 0$. The bottom two lines are indifference prices in the first non-degenerate case with $x^*>0$.}
\label{fig:pxR}
\end{figure}

An explanation of this phenomena is as follows. Consider an agent with positive cash wealth and zero endowment of the risky asset. This agent consumes at rate $\beta x/R$; in particular, as the parameter $R$ increases, the agent consumes more slowly. The introduction of a small endowment will not change this result, and in general, an increase in the parameter $R$ postpones the time at which the critical ratio reaches $z^*$. (Although $z^*$ depends on $R$ also, this is a secondary effect.) Since the endowed asset is appreciating, on average, by the time the agent chooses to start selling the asset, it will be worth more. The total effect is to make the indifference price increasing in $R$. Similarly, the indifference price $p(1,1,\theta)$ and the unit indifference price $p(1,1,\theta)/\theta$ as functions of $\theta$ are not necessarily monotone in risk aversion.

Finally, we consider the impact of the illiquidity assumption. We do this by considering the value function of our agent who cannot buy the endowed asset and comparing it with the value function of an otherwise identical agent, but who can both buy and sell the endowed asset with zero transaction costs. Suppose parameters are such that we are in the second case of Theorem~\ref{thm:4cases}.
In the illiquid market, where $Y$ is only allowed for sale, Theorem~\ref{thm:maincase} proves the value function is
\begin{equation}
V_I(x, y, \theta, 0) = \frac{x^{1-R}}{1-R} g\left(\frac{y\theta}{x}\right) =
\displaystyle \sup_{(C,\Theta)} \mathbb{E} \left[\int_0 ^{\infty} e^{-\beta t} \frac{C_t ^{1-R}}{1-R}dt \right],
\label{eq:vi}
\end{equation}
where the newly introduced subscript $I$ stands for the value function in the illiquid market, in which the asset can only be sold.

In a liquid market such that $Y$ can be dynamically traded, wealth evolves as $dX_t = -C_t dt + \Pi_t dY_t/Y_t$. Here $(\Pi)_{t \geq 0}$ represents the portfolio process. We suppose the agent is endowed with $\Theta_0$ units of $Y$ initially and is constrained to keep $X$ positive. This is Merton's model and we know the optimal strategy is to keep a constant fraction of wealth in the risky asset. The initial endowment therefore only changes initial wealth and the value function is
\begin{equation}
V_L(x, y, \theta, 0) = \displaystyle \sup_{(C,\Pi)} \mathbb{E} \left[\int_0 ^{\infty} e^{-\beta t} \frac{C_t ^{1-R}}{1-R}dt \right] =
\frac{(x + y \theta)^{1-R}}{1-R}
\left[\frac{\beta}{R} - \frac{\alpha^2(1-R)}{2\sigma^2R^2}\right]^{-R},
\label{eq:vl}
\end{equation}
where the subscript $L$ stands for the value function in the liquid market.

Now we consider the cost of illiquidity.

\begin{defn}

The cost of illiquidity, denoted $p^* = p^*(x,y,\theta)$ is the solution to
\begin{equation}
V_{L} (x - p^*, y, \theta, t) = V_{I} (x,y,\theta, t).
\label{eq:defill}
\end{equation}
and represents the amount of cash wealth the agent who can only sell the risky asset would be prepared to forgo, in order to be able to trade the risky asset with zero transaction costs.
\label{def:cl}
\end{defn}

Equating (\ref{eq:vi}) and (\ref{eq:vl}), we can solve for $p^{*}$ to obtain
\begin{equation}
p^{*}(x,y,\theta) = x \left[1 + \frac{y\theta}{x} - g\left(\frac{y\theta}{x}\right)^{\frac{1}{1 - R}}
\left(\frac{\beta}{R} - \frac{\alpha^2(1-R)}{2\sigma^2R^2}\right)^{\frac{R}{1-R}} \right].
\label{eq:pstar}
\end{equation}
Consider (\ref{eq:pstar}) when $\theta = 0$, where the investor is not endowed any units of $Y$ initially, we have
\[
p^{*}(x,y,0) = x \left[1 -  \left(\frac{\beta}{R} - \frac{\alpha^2(1-R)}{2\sigma^2R^2}\right)^{\frac{R}
{1 - R}}
g(0)^{\frac{1}{1-R}} \right] = x \left[1 - \left(1 -  \frac{\epsilon^2 (1-R)}{2 \delta^2 R}  \right)^
{\frac{R}{1-R}} \right] > 0.
\]

Suppose $R<1$, $0 < \epsilon < \frac{\delta^2}{2} R + \frac{1}{1-R}$ and
$\epsilon < \delta^2 R$, so that $z^*$ is finite.
Figure~\ref{fig:costofilliquiditytheta} plots $p^*(1,1,\theta)$ for $\theta \in [0, 10]$. Notice that $p^*$ decreases initially, has a strictly positive minimum near 0.95 and
then increases, before becoming linear beyond $\theta = z^*$. Clearly, whatever the initial endowment of the agent, she has a smaller set of admissible strategies than an agent who can trade dynamically, and the cost of liquidity is strictly positive. For small initial endowments the agent would like to increase the size of her portfolio of the risky asset, and the smaller her initial endowment the more she would like to purchase at time zero. Hence the cost of illiquidity is decreasing in $\theta$ for small $\theta$. However, for large $\theta$, the agent would like to make an initial transaction (to reduce the ratio of wealth held in the risky asset to cash wealth to below $z^*$), and indeed since she is free to do so, her optimal strategy involves such a transaction at time zero. Hence for large wealth the cost of liquidity is proportional to $(x + y\theta)$, and hence is increasing in $\theta$. For this reason, the cost of illiquidity is a U-shaped function of $\theta$.

\begin{figure}
\begin{centering}
\includegraphics[scale=0.29]{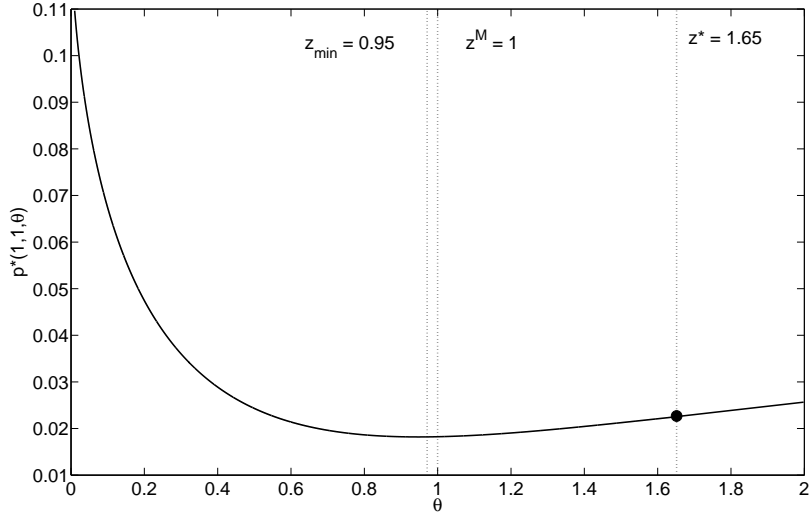}
\par\end{centering}

\caption{Cost of illiquidity $p^{*}(1,1,\theta)$ as $\theta$ varies. Parameters are $\epsilon = 1$, $\delta = 2$ and $R=0.5$. Here, we fix $x_0 = y_0 = 1$ and $\theta \in [0,1]$. For the corresponding Merton problem with dynamic trading in $Y$ we have that it is optimal to invest a constant fraction $z^M =  \frac{\epsilon}{\delta^2 R - \epsilon}$ in the risky asset. Recall Remark~\ref{rem:mertonline} and observe that $z^M \leq z^*$.}
\label{fig:costofilliquiditytheta}
\end{figure}

%% file: appendix140815.tex
\section{Properties of $n$}
\label{app:propertiesofna}

Recall the definitions of $m$ and $\ell$ and the differential equation
(\ref{eqn:node}) for $n$, and also the definitions of $\dummyq_\ell$,
$\dummyq_m$, $\dummyq_n$ and $\dummyq^*$. Define $\tilde{\dummyq} = \inf \{
\dummyq>0 : (1-R)n(\dummyq) \ge (1-R)\ell(\dummyq) \} \wedge 1$.
 Note that
$m\left(0\right)=1=\ell\left(0\right)$ and
$m\left(1\right)=1-\epsilon(1-R)+\delta^{2}R\left(1-R\right)/2
=\ell\left(1\right)$. The concave function $\ell$ is positive on
$\left(0,1\right)$ if $\ell(1) =
1-\epsilon(1-R)+\delta^{2}R\left(1-R\right)/2\geq0$.

\begin{lem}
\label{lem:n}
\begin{enumerate}

\item Define $\Phi$ via
\[ \Phi(\chi) = \chi^2 - (1-R) \left( \frac{\delta^2}{2} - \epsilon + \frac{1}{R}
\right) \chi - \epsilon \frac{(1-R)^2}{R}. \]
Then for $R \in (0,1)$,
$n'(0)$ is the smaller root of $\Phi(\chi)=0$
and for $R \in (1, \infty)$, $n'(0)$
is the larger root.

\item For $\dummyq \in (0,\dummyq_n \wedge \tilde{\dummyq})$, $n'(\dummyq)>0$
if and only if $n(\dummyq)<m(\dummyq)$, similarly $n'(\dummyq)=0$ if and only
if $n(\dummyq)=m(\dummyq)$.

\item If $\ell(1) \geq 0$ then $\tilde{\dummyq}=\dummyq_n = \dummyq_\ell = 1$.

\item If $\ell(1) < 0$ then $\tilde{\dummyq} = \dummyq_n = \dummyq_\ell <
\dummyq^*$.

\item If $0 \leq \dummyq^* < 1$ then $\dummyq^* > \epsilon/\delta^2 R$ and $(1
- R)m$ is increasing on $(\dummyq^*,1)$.

\end{enumerate}

\end{lem}

\begin{proof}
(1) From the expression (\ref{eqn:node}) and l'H\^{o}pital's rule,
$n'(0)=\chi$ solves
\[ \chi = \frac{1-R}{R} - \frac{\delta^2}{2} \frac{(1-R)^2}{R}
\frac{1}{(1-R)(\frac{\delta^2}{2} - \epsilon)-\chi}, \]
or equivalently $\Phi(\chi)=0$. Further $\ell'(0) =  (1-R) \left(
\frac{\delta^2}{2} - \epsilon \right)$ and
\[ \Phi \left( (1-R) \left( \frac{\delta^2}{2} - \epsilon \right) \right) = -
\frac{\delta^2}{2} \frac{(1-R)^2}{R} < 0. \]
For $R<1$, we have $n'(0)< \ell'(0)$ by hypothesis, so that $n'(0)$ is the smaller root
of $\Phi$. For $R>1$, we have $n'(0)> \ell'(0)$ by hypothesis and $n'(0)$ is the larger
root of $\Phi$.

(2) This follows immediately from the expression for $n'(\dummyq)$.

(3)
Suppose $R<1$. Since $n'(0)< \ell'(0)$ we have
$\tilde{\dummyq}>0$. Notice that if $0<n(\dummyq)< \ell(\dummyq)$ and $\ell(\dummyq)-n(\dummyq)$ is
sufficiently small, then $n'(\dummyq) < \ell'(\dummyq)$. Hence $\tilde{\dummyq} \geq \dummyq_n$.
Further, if
$n\left(\dummyq\right)<\ell\left(\dummyq\right)-\phi$ for some $\phi>0$ on some
interval
$\left[\underline{\dummyq},\overline{\dummyq}\right]\subset\left(0,1\right)$, then
$n'\left(\dummyq\right)/n\left(\dummyq\right)$ is bounded below by a constant on
that interval and provided $n\left(\underline{\dummyq}\right)>0$ it follows
that $n\left(\overline{\dummyq}\right)>0$ also. Hence, if $\ell$ is
positive on $[0,1)$ then so is $n$ and
$\dummyq_n=1$. For $R>1$, we have $n'(0) > \ell'(0)$ and the result follows via
a similar argument.

(4)
Suppose $R<1$.
The same argument as above gives that $\tilde{\dummyq}=\dummyq_n = \dummyq_\ell$ and now
these quantities are less than one. Clearly $\dummyq_m< \dummyq_\ell$, and $m$ is
decreasing on $(0,\dummyq_m)$. We cannot have $\dummyq^* \leq \dummyq_m$ for then $n'(\dummyq^*)
- m'(\dummyq^*) > 0$ and $n(\dummyq^*)-m(\dummyq^*)=0$ contradicting the minimality of
$\dummyq^*$, nor can we have $\dummyq_m < \dummyq^* \leq \dummyq_\ell$ for on this region $m<0\leq
n$.

(5) We can only have $\dummyq^*<1$ if $m(1) > 0$ and $(1 - R) m'(1)>0$. For $R<1$
we must
have $n'(\dummyq^*) = 0 < m'(\dummyq^*)$. But $m$ has a minimum at
$\epsilon/\delta^2 R$, so $\dummyq^*> \epsilon/\delta^2 R$. For $R>1$, we must
have $n'(\dummyq^*) = 0 > m'(\dummyq^*)$. But $m$ has a maximum at
$\epsilon/\delta^2 R$, so $\dummyq^*> \epsilon/\delta^2 R$.

\end{proof}

\begin{proof}[Proof of Proposition~\ref{prop:crossings}.]
(1) Note that $\Phi(m'(0)) = (1-R)^2 \delta^2 \epsilon/2$. Then, if
$\epsilon<0$ we have $n'(0)<m'(0)$ for $R<1$ and $\dummyq^*=0$. Otherwise, for $R>1$,
we have $n'(0)>m'(0)$ and $\dummyq^*=0$. If $\epsilon = 0$ then
$n'(0)=m'(0)$ and
more care is needed.

Consider $R<1$. Since $\epsilon \leq 0$, $m$ is increasing. Suppose
$n\left(\hat{\dummyq}\right)>m\left(\hat{\dummyq}\right)$ for some $\hat{\dummyq}$
in $\left[0,1\right].$ Let $\underline{\dummyq}=\sup\left\{
\dummyq<\hat{\dummyq}:n\left(\dummyq\right)=m\left(\dummyq\right)\right\} $. Then on
$\left(\underline{\dummyq},\hat{\dummyq}\right)$ we have
$n'\left(\dummyq\right)<0<m'\left(\dummyq\right)$ and
$m\left(\hat{\dummyq}\right)-n\left(\hat{\dummyq}\right)
=m\left(\underline{\dummyq}\right)-n\left(\underline{\dummyq}\right)
+\int_{\underline{\dummyq}}^{\hat{\dummyq}}
[m'\left(y\right)-n'\left(y\right)]dy>0$,
a contradiction. 

For $R>1$, the only difference is that $m$ is decreasing
given  $\epsilon \leq 0$ and $n'(0)>m'(0)$.

(2) Consider first $R<1$ and suppose that
$0 < \epsilon< \min \{ \delta^{2}R , \frac{\delta^2}{2}R + \frac{1}{1-R} \}$.
Then $m'\left(1\right)>0$ and $m(1)>0$. Since $\epsilon>0$
we have $n'\left(0\right)>m'\left(0\right)$ and $n-m$ is positive
at least initially. Write
$n\left(\dummyq\right)=m\left(\dummyq\right)
+\delta^{2}\left(1-R\right)\dummyq b\left(\dummyq\right)/2$.
Then $n\left(\dummyq\right)\leq\ell\left(\dummyq\right)$ implies
$b\left(\dummyq\right)\leq1-\dummyq$.

Suppose $b\left(\dummyq\right)>0$ for all $\dummyq\in\left(0,1\right)$. Then
$n\left(\dummyq\right)\geq m\left(\dummyq\right)$ and $n'\left(\dummyq\right)<0$
so that $n\left(\dummyq\right)\geq n\left(1\right)=m\left(1\right)$ and
\begin{eqnarray*}
m\left(1\right)&=&m\left(\dummyq\right)-\left(1-\dummyq\right)\left(1-R\right)
\left(\epsilon-\delta^{2}R\right)
-\left(1-\dummyq\right)^{2}\delta^{2}R\left(1-R\right)/2 \\
& > & m\left(\dummyq\right)
+\phi\left(1-\dummyq\right)\delta^{2}\left(1-R\right)\dummyq/2,
\end{eqnarray*}
for $\dummyq>\epsilon/\delta^{2}R$ and $\phi<(\delta^{2}R-\epsilon) \min \{ \frac{2}{\delta^2}, \frac{R}{\epsilon} \}$. For
such $\dummyq$, $b\left(\dummyq\right)>\phi\left(1-\dummyq\right)$. Hence
\[
\frac{n'\left(\dummyq\right)}{n\left(\dummyq\right)}
=-\frac{1-R}{R}\frac{b\left(\dummyq\right)}
{\left(1-\dummyq\right)\left(1-\dummyq-b\left(\dummyq\right)\right)}\leq
-\frac{1-R}{R}\frac{\phi}{\left(1-\dummyq\right)\left(1-\phi\right)}
\]
and we must have $n'\left(1-\right)=-\infty$ contradicting the
fact that $n\left(\dummyq\right)\leq\ell\left(\dummyq\right)$. It follows
that we must have $b\left(\dummyq\right)=0$ for some $\dummyq\in\left(0,1\right)$.
At this point $n$ crosses $m$. Note that this crossing point is
unique: at any crossing point
$m'\left(\dummyq\right)>0=n'\left(\dummyq\right)$,
so that all crossings of $0$ in $\left(0,1\right)$ by $n-m$ are
from above to below.

For $R>1$, we have $m'(1)<0$ and $m(1)>0$. Since
$\epsilon > 0$, we have $n'(0)<m'(0)$ and $n-m$ is negative initially.
Let $n(\dummyq) = m(\dummyq) + \delta^2 (1-R)\dummyq b(\dummyq)/2$. Then
$n(\dummyq) \geq \ell(\dummyq)$ implies
$b(\dummyq) \leq 1 - \dummyq$. Suppose $b(\dummyq)>0$ for all $\dummyq \in (0,1)$, then it leads to the
same contradiction for $R<1$. It follows that $b(\dummyq) = 0$ for some $\dummyq \in (0,1)$,
where $n$ crosses $m$. At any crossing point $m'(\dummyq) < 0 = n'(\dummyq)$, so that $n$
crosses $m$ from below.

(3)
$\epsilon \geq \delta^{2}R$ and if $R<1$, $\epsilon < \frac{\delta^2}{2}R + \frac{1}{1-R}$.

Consider first $R<1$. Since $\epsilon>0$ we have that
$n'\left(0\right)>m'\left(0\right)$
and $n>m$ in a neighbourhood to the right of zero. Further, $m$
is decreasing and there are no solutions of $n=m$ since at any solution
we must have that $0=n'<m'<0$.

For $R>1$, we have $m$ is increasing and $n'(0)<m'(0)$. There are no solutions
of $n = m$ in that at any solution we should have $0 = n' > m' >0$.

(4)
$R<1$ and $\epsilon \geq \frac{\delta^2}{2}R + \frac{1}{1-R}$

Then $m\left(1\right)\leq0$. Since $m$ is decreasing at least until
it hits zero, and since $n'=0$ at a crossing point we cannot have
that $n$ crosses $m$ before it hits zero.
\end{proof}

\begin{proof}[Proof of Proposition ~\ref{prop:NWw}]
(1) $N$ solves
\[
N'\left(\dummyq\right)=\frac{\frac{1}{2}\delta^{2}\left(1-R\right)^{2}\dummyq N
\left(\dummyq\right)}{\ell\left(\dummyq\right)
-N\left(\dummyq\right)^{-1/R}\left(1-\dummyq\right)^{1-1/R}}
\]
and $N$ is strictly increasing for $R<1$. Otherwise, it is decreasing for $R>1$.
$W$ solves
\begin{equation}
W'\left(v\right)=\frac{\ell\left(W\left(v\right)\right)
-v^{-1/R}\left(1-W\left(v\right)\right)^{1-1/R}}{\frac{1}{2}\delta^{2}
\left(1-R\right)^{2}vW(v)}
\label{eqn:Wode}
\end{equation}

(2) Follows from (\ref{eq:885}) and  (\ref{eqn:Wode}).

(3) Consider first $R<1$. On $(0,\dummyq^*)$ we have $n(\dummyq) > m(\dummyq)$ and then $\ell(\dummyq)-n(\dummyq) <
\ell(\dummyq) - m(\dummyq) = \dummyq(1-\dummyq) \delta^2 (1-R)/2$.
Then $v^{-1/R}(1-W(v))^{1-1/R} = n(W(v))$ and
\[ v(1-R)W'(v) = \frac{\ell(W(v)) -
n(W(v))}{\frac{\delta^2}{2}(1-R)W(v)} < 1-W(v) \]
It follows that $w'(v) = (1-R)W(v) + v(1-R)W'(v) < 1 - RW(v)$. At $\dummyq^*$,
$n(\dummyq^*)=m(\dummyq^*)$ and the inequality becomes an equality throughout.

For $R>1$, we have $n(\dummyq)<m(\dummyq)$ on $(0,\dummyq^*)$ and
$\ell(\dummyq) - n(\dummyq) > \ell(\dummyq) - m(\dummyq)= \dummyq(1-\dummyq) \delta^2 (1-R)/2$.
Then again $v(1-R)W'(v) < 1-W(v)$ and $w'(v)  < 1 - RW(v)$ with equality at $h^*$.

Note that since $W$ is non-negative, $1-RW(h) \leq 1$.
\end{proof}

\section{The martingale property of the value function}
\label{app:mg}

\begin{proof}[Proof of Lemma~\ref{lem:mg}.]
First we want to show
the the local martingale
\[
N_{t}^{3}=\int_{0}^{t}\eta Y_s G_{y}(X^*_s, Y_s, \Theta^*_s,s)
dB_{s}
\]
is a martingale. This will follow if, for example,
\begin{equation}
\mathbb{E}\int_{0}^{t}\left(Y_s G_{y}(X^*_s, Y_s, \Theta^*_s,s)
\right)^{2}ds<\infty\label{eq:4}
\end{equation}
for each $t>0$.
From the form of the value function (\ref{eqn:Gdef}), we have
\begin{equation}
\label{eqn:Gy}
y G_y(x,y,\theta,s)
= e^{-\beta t}\frac{x^{1-R}}{1-R}zg'\left(z\right)
= G\left(x,y,\theta,t\right)
\frac{zg'\left(z\right)}{g\left(z\right)}
\leq (1-R)G\left(x,y,\theta,t\right)
\end{equation}
where we use that
$\frac{zg'\left(z\right)}{g\left(z\right)}
=\frac{w\left(h\right)}{h} = (1-R)W(h)$ and $0 \leq W(h) \leq 1$.

Define a process $\left(D_{t}\right)_{t\geq0}$ by
$ 
D_{t}=\ln G\left(X^*_{t},Y_{t},\Theta^*_{t},t\right)
$. 
Then $D$ solves
\begin{eqnarray*}
D_{t}-D_{0} &=&  \int_{0}^{t}\frac{1}{G}
\left(G_t - C^*_sG_{x}+\alpha
Y_sG_{y}+\frac{1}{2}\eta^{2}Y_s^{2}G_{yy}\right)ds\\
&& \hspace{5mm} +
\int_{0}^{t}\frac{1}{G}\left(G_{\theta}-Y_sG_{x}\right)d\Theta_s
+ \int_{0}^{t}\frac{1}{G}\eta Y_s G_{y}dB_{s}
- \int_{0}^{t}\frac{1}{2G^{2}}\eta^{2}Y_s^{2}G_{y}^{2}ds\\
&= &
-\int_{0}^{t}\frac{e^{-\frac{\beta}{R}s}}{1-R}\frac{1}{G}G_{x}^{\frac{R-1}{R}}ds
+\int_{0}^{t}\frac{1}{G}\eta
Y_s G_{y}dB_{s}
- \int_{0}^{t}\frac{1}{2G^{2}}\eta^{2}Y_s^{2}G_{y}^{2}ds.
\end{eqnarray*}
It follows that the candidate value function along the optimal trajectory has the
representation
\begin{equation}
\label{eqn:G0}
G\left(X^*_t,Y_t,\Theta^*_t,t\right)=
G\left(X^*_{0},y_{0},\Theta^*_{0},0\right) \exp\left\{
- \int_{0}^{t}\frac{e^{-\frac{1}{R}\beta
s}}{1-R}\frac{1}{G}G_{x}^{\frac{R-1}{R}}ds\right\} H_t
\end{equation}
where $H=\left(H_{t}\right)_{t\geq0}$ is the exponential martingale
\[ 
H_{t}= \sE \left( \frac{\eta Y_s G_y}{G} \circ B \right)_t :=
\exp\left\{
\int_{0}^{t}\frac{1}{G}\eta
Y_{s}G_{y}dB_{s}-\int_{0}^{t}\frac{1}{2G^{2}}
\eta^{2}Y_{s}^{2}G_{y}^{2}ds\right\}
.  
\]
Note that (\ref{eqn:Gy}) implies
$\frac{1}{G}\eta yG_{y} \leq\eta (1-R)$,
so that $H$ is indeed a martingale, and not merely a local martingale.

From (\ref{eqn:Gy}) and (\ref{eqn:G0}), we have
\begin{eqnarray*}
\left(yG_{y}\right)^{2}
& = &
G\left(X_{0},y_{0},\Theta_{0},0\right)^{2} \left(\frac{zg'
\left(z\right)}{g\left(z\right)}\right)^{2}\times\exp\left\{
-2\int_{0}^{t}\frac{e^{-\frac{1}{R}\beta
s}}{(1-R)}\frac{1}{G}G_{x}^{\frac{R-1}{R}}ds\right\} H_t^2
\\
& \leq  &
G\left(X_{0},y_{0},\Theta_{0},0\right)^{2} (1-R)^2
H_t^2.
\end{eqnarray*}
But
\[ H_t^2 =
\sE \left( \frac{2}{G}\eta
Y_{s}G_{y} \circ B \right)_{t}
\exp\left\{ \int_{0}^{t}\frac{1}{G^{2}}\eta^{2}
Y_{s}^{2}G_{y}^{2}ds\right\}
\leq
\sE \left( \frac{2}{G}\eta
Y_{s}G_{y} \circ B \right)_{t}
e^{(1-R)^2 \eta^2 t}. \]
Hence $\E[H_t^2] \leq e^{(1-R)^2 \eta^2 t}$ and it follows that
(\ref{eq:4}) holds
for every $t$, and hence that the local martingale
$N_{t}^{3}=\int_{0}^{t}\eta yG_{y}dB_{s}$ is a martingale
under the optimal strategy.

(ii) Consider
$ \int_{0}^{t}\frac{e^{-\frac{1}{R}\beta
s}}{1-R}\frac{1}{G}G_{x}^{\frac{R-1}{R}}ds $.
To date we have merely argued that this function is increasing in
$t$. Now we want to argue that it grows to infinity at least linearly.
By (\ref{eqn:Gdef}), we have
\begin{eqnarray*}
\frac{e^{-\frac{1}{R}\beta
t}}{1-R}\frac{1}{G}G_{x}^{\frac{R-1}{R}} & = &
\frac{\left[g\left(z\right)-\frac{1}{1-R}zg'\left(z\right)\right]^{\frac{R-1}{R}}}
{g\left(z\right)}
=  \frac{\left[h-\frac{1}{1-R}w\left(h\right)\right]^{\frac{R-1}{R}}}{h} \\
& = & (1-W(h))^{1-1/R} h^{-1/R} = n(W(h)) \geq \min \{ 1,n(W(h^*)) \} > 0.
\end{eqnarray*}
Hence from (\ref{eqn:G0}) there exists a constant $k>0$ such that
\[ 0 \leq (1-R) G(X^*_t,Y_t, \Theta^*_t, t) \leq (1-R) G(x_0, y_0, \theta_0,0)
e^{- k t} H_t \rightarrow 0 \]
and then $G \rightarrow 0$ in $L^1$, as required.
\end{proof}

\begin{proof}[Proof of Lemma~\ref{lem:mgc3}.]
This follows exactly as in the proof of Lemma~\ref{lem:mg}.

\end{proof}

\section{Extension to $R>1$}
\label{app:R>1}

\begin{proof}[Verification Lemmas for the case $R>1$.]
It remains to extend the proofs of the verification lemmas to the case $R>1$.
In particular we need to show that the candidate value function is an upper bound
on the value function. The main idea is taken from Davis and
Norman~\cite{Davis}.

Suppose $G\left(x,y,\theta,t\right)$ is the candidate value function.
Consider for $\varepsilon>0$,
\begin{equation} \widetilde{V}_{\varepsilon} (x,y,\theta,t) =
\widetilde{V}\left(x,y,\theta,t\right)=G\left(x+\varepsilon,y,\theta,t\right)
\label{eq:degvr}
\end{equation}
and $\widetilde{M}_t = \widetilde{M}_t(C, \Theta)$ given by
\[
\widetilde{M}_t = \int_0 ^t e^{-\beta s} \frac{C_s ^{1-R}}{1-R}ds +
\widetilde{V}\left(X_t, Y_t, \Theta_ t, t\right),
\]
Then,
\begin{eqnarray*}
\widetilde{M}_{t}-\widetilde{M}_{0} & = &
\int_{0}^{t}\left[e^{-\beta s}\frac{C_s^{1-R}}{1-R}
- C_s\widetilde{V}_{x}+\alpha Y_s\widetilde{V}_{y}+\frac{1}{2}\eta^{2}Y_s^{2}
\widetilde{V}_{yy}+\widetilde{V}_{t}\right]ds \\
& & + \int_{0}^{t}\left(\widetilde{V}_{\theta} -
Y_s\widetilde{V}_{x}\right)d\Theta_s \\
& & + \sum_{\substack{0 \leq s \leq t}}
\left[\widetilde{V}(X_s, Y_s, \Theta_s, s) - \widetilde{V}(X_{s-}, Y_{s-},
\Theta_{s-}, s-) - \widetilde{V}_x (\bigtriangleup X)_s
- \widetilde{V}_\theta (\bigtriangleup \Theta)_s \right] \\
& & +  \int_{0}^{t}\eta Y_s\widetilde{V}_{y}dB_{s}  \\
& = & \widetilde{N}_{t}^{1}+\widetilde{N}_{t}^{2}+\widetilde{N}_{t}^{3}
+\widetilde{N}_{t}^{4}.
\end{eqnarray*}
Lemma~\ref{deg1:inequal} (in the case $\epsilon \leq 0$ and otherwise
Lemma~\ref{lem:operator} or
Lemma~\ref{lem:2ndoperator})
implies
$\widetilde{N}_{t}^{1} \leq 0$ and
$\widetilde{N}_{t}^{2} \leq 0$. The concavity of $\widetilde{V}(x+y \chi,y,\theta -
\chi,s)$ in $\chi$ (either directly if $\epsilon \leq 0$, or using
Lemma~\ref{lem:concave} or Lemma~\ref{lem:2ndconcave})
implies $(\Delta
\widetilde{N}^3) \leq 0$.

Now define
stopping times $\tau_{n}=\inf\left\{
t\geq0:\int_{0}^{t}\eta^{2}Y_s^{2}\widetilde{V}_{y}^{2}ds\geq n\right\}$.
It follows from (\ref{eqn:Gy}) that $y\widetilde{V}_{y}$ is bounded and hence
$\tau_n \uparrow \infty$.
Then
the local martingale $(\widetilde{N}_{t \wedge \tau_n}^{4})_{t \geq 0}$ is a
martingale and
taking expectations we
have $\mathbb{E}\left(\widetilde{M}_{t\wedge\tau_{n}}\right)\leq\widetilde{M}_{0}$,
and hence
\[
\mathbb{E}\left(\int_{0}^{t\wedge\tau_{n}}e^{-\beta s}
\frac{C_{s}^{1-R}}{1-R}ds+
\widetilde{V}\left(X_{t\wedge\tau_{n}},
Y_{t\wedge\tau_{n}},\Theta_{t\wedge\tau_{n}},t\wedge\tau_{n}\right)\right)
\leq\widetilde{V}\left(x_{0},y_{0},\theta_{0},0\right). \]

In the case $\epsilon \leq 0$,
(\ref{eq:vdeg1}) and (\ref{eq:degvr}) imply
\[
\widetilde{V}\left(x,y,\theta,t\right)
= e^{-\beta t}\frac{\left(x
+\varepsilon\right)^{1-R}}{1-R} \left(1 +
\frac{y \theta }{x
+\varepsilon}\right)^{1-R} \left(\frac{R}{\beta}\right)^R
\geq
e^{-\beta t}\frac{\left(x
+\varepsilon\right)^{1-R}}{1-R} \left(\frac{R}{\beta}\right)^{R} \geq
\frac{\varepsilon^{1-R}}{1-R} \left(\frac{R}{\beta}\right)^R. \]
Thus $\widetilde{V}$ is bounded,
\(
\underset{n\rightarrow\infty}{\lim}\mathbb{E}\widetilde{V}
\left(X_{t\wedge\tau_{n}},Y_{t\wedge\tau_{n}},
\Theta_{t\wedge\tau_{n}},t\wedge\tau_{n}\right)
= \mathbb{E}\left[\widetilde{V}\left(X_{t},Y_{t},\theta_{t},t\right)\right],
\)
and
\[
\widetilde{V}\left(x_{0},y_{0},\theta_{0},0\right)\geq
\mathbb{E}\left(\int_{0}^{t}e^{-\beta s}\frac{C_{s}^{1-R}}{1-R}ds\right)
+\mathbb{E}\left[\widetilde{V}\left(X_t, Y_t, \Theta_t, t\right)\right].
\]
Similarly,
\[ \widetilde{V}(x,y,\theta,t) \geq e^{-\beta t} \frac{\varepsilon^{1-R}}{1-R}
\left(\frac{R}{\beta}\right)^R \]
and hence $\mathbb{E}\left[\widetilde{V}\left(X_t, Y_t, \Theta_t,
t\right)\right] \rightarrow 0$.
Then letting $t \to \infty$ and applying the monotone convergence theorem,
we have
\[ 
\widetilde{V}_{\varepsilon}\left(x_{0},y_{0},\theta_{0},0\right) =
\widetilde{V}\left(x_{0},y_{0},\theta_{0},0\right)\geq
\mathbb{E}\left(\int_{0}^{\infty}e^{-\beta s}\frac{C_{s}^{1-R}}{1-R}ds\right)
\] 
Finally let $\varepsilon \to 0$. Then $V \leq
\lim_{\varepsilon \downarrow 0} \widetilde{V} = G$. Hence, we have $V \leq G$.

The two non-degenerate cases are very similar, except that now
from (\ref{eqn:Gdef}) and (\ref{eq:degvr}),
\[
\widetilde{V}\left(x,
y,\theta,t\right)
= e^{-\beta t}
\frac{\left(x
+\varepsilon\right)^{1-R}}{1-R}
g\left(\frac{y \theta}{x +\varepsilon}\right) \geq
e^{-\beta t} \frac{\varepsilon^{1-R}}{1-R} \left(\frac{R}{\beta}\right)^R.
\]
where we use that for $R>1$, $g$ is decreasing with
$g\left(0\right) = (\frac{R}{\beta})^R > 0$. Hence
$\widetilde{V}$
is bounded, and the argument proceeds as before.

\end{proof}